\documentclass[12pt]{amsart}
\usepackage{amsmath}
\usepackage{amssymb}
\usepackage{amscd}
\usepackage{epsfig}
\usepackage{color}
\usepackage{mathrsfs}
\usepackage{bbding}

\newtheorem{theorem}{Theorem}[section]
\newtheorem{proposition}{Proposition}[section]
\newtheorem{lemma}{Lemma}[section]
\newtheorem{corollary}{Corollary}[section]

\theoremstyle{definition}
\newtheorem{definition}{Definition}[section]
\newtheorem{remark}{Remark}[section]
\newtheorem{example}{Example}[section]

\newcommand{\fWAP}{\mathsf{WAP}}
\newcommand{\fSAP}{\mathsf{SAP}}

\DeclareMathOperator{\vol}{vol}

\DeclareMathOperator{\dens}{dens}

\DeclareMathOperator{\supp}{supp}

\newcommand{\cD}{\mathscr{D}}

\newcommand{\cM}{\mathcal{M}}

\newcommand{\cS}{\mathscr{S}}

\newcommand{\cMS}{\mathcal{S}_{\widehat{\mathcal{M}}}}

\newcommand{\NN}{\mathbb{N}}

\newcommand{\ZZ}{\mathbb{Z}}

\newcommand{\RR}{\mathbb{R}}
\newcommand{\CC}{\mathbb{C}}

\newcommand{\cSp}{\cS^{\prime}}
\newcommand{\Cc}{C_c(\RR^d)}
\newcommand{\Cu}{C_u(\RR^d)}

\newcommand{\dd}{\,\mathrm{d}}

\newcommand{\ii}{\ts\mathrm{i}\ts}

\newcommand{\ts}{\hspace{0.5pt}}

\newcommand{\id}{\mathtt{id}}

\newcommand{\n}{\text{-}}

\newcommand{\norm}[1]{\Vert #1 \Vert}
\newcommand{\ds}{\displaystyle}

\newcommand{\econv}{\circledast}

\begin{document}

\title{Diffraction theory and almost periodic distributions}

\author{Nicolae Strungaru}
\address{Department of Mathematical Sciences, MacEwan University \\
10700 “ 104 Avenue, Edmonton, AB, T5J 4S2;\\
and \\
Department of Mathematics\\
Trent University \\
Peterborough, ON
and \\
Institute of Mathematics ``Simon Stoilow'' \\
Bucharest, Romania}
\email{strungarun@macewan.ca}
\urladdr{http://academic.macewan.ca/strungarun/}

\author{Venta Terauds}
\address{Fakult\"at f\"ur Mathematik, Universit\"at Bielefeld, Postfach 100131, 33501 Bielefeld, Germany}
\email{terauds@math.uni-bielefeld.de}

\begin{abstract} We introduce and study the notions of translation bounded tempered distributions, and autocorrelation for a tempered distrubution. We further introduce the spaces of weakly, strongly and null weakly almost periodic tempered distributions and show that for weakly almost periodic tempered distributions the Eberlein decomposition holds. For translation bounded measures all these notions coincide with the classical ones. We show that tempered distributions with measure Fourier transform are weakly almost periodic and that for this class, the Eberlein decomposition is exactly the Fourier dual of the Lesbegue decomposition, with the Fourier-Bohr coefficients specifying the pure point part of the Fourier transform. We complete the project by looking at few interesting examples.
\end{abstract}

\maketitle

\section{Introduction}

Diffraction is a key tool in the study of non-periodic long-range order, and the mathematical theory of almost periodic measures and almost periodic functions that underpins this is well-developed \cite{ARMA}; see \cite{MoSt} for a recent overview of the theory. Physical structures that possess a diffraction pattern are generally represented by translation bounded measures. Any translation bounded measure has at least an autocorrelation and hence a diffraction \cite{hof95}; see \cite{book} for a comprehensive introduction and overview of the subject.

Essential to the study of diffraction is the theory of the Fourier transform in the setting of tempered distributions. In many of the established results, the proofs rely on the fact that the original translation bounded measure is tempered as a distribution and that the autocorrelation is positive definite. Therefore, it is natural to ask whether some results about the diffraction of translation bounded measures can be extended to tempered distributions.
In \cite{ter} it was proven that given any diffraction measure supported on an infinite point set, one may construct non-measure tempered distributions with that diffraction. Via this inverse problem technique \cite{lenzmoo} one obtains plenty of non-trivial tempered distributions with a diffraction, but not a general theory in which they fit.

In this paper we consider when and how one may define autocorrelation (and hence diffraction) for a tempered distribution, and whether known results about the diffraction of translation bounded measures can be extended to tempered distributions.
We note that the constructions given in \cite{ter} all have (translation bounded) measure autocorrelation. However, the autocorrelation of a tempered distribution need not, in general, be a measure. As long as a tempered distribution is positive definite, the Bochner Schwartz theorem ensures that its Fourier transform is a positive tempered measure and thus can be understood as a diffraction.

The concept of translation boundedness for tempered distributions is formally defined in the second section. Here we consider some examples of such distributions and show that they have some properties analogous to those of translation bounded measures. In fact a translation bounded measure is always a translation bounded tempered distribution, although of course the converse does not hold.

In the third section, we give a formal definition of autocorrelation for a tempered distribution, and some sufficient conditions for existence of an autocorrelation. The properties of positive definite tempered distributions are considered in Section 4, and in Section 5 we introduce the concept of weak almost periodicity for tempered distributions. We here construct the Eberlein decomposition for tempered distributions: as for measures, a weakly almost periodic tempered distribution may be uniquely decomposed into the sum of a null weakly almost periodic and a strongly almost periodic tempered distribution.

A tempered distribution whose Fourier transform is a measure is necessarily weakly almost periodic. In Section 6, we show that for such a distributions, the Fourier transform carries the Eberlein decomposition of the distribution precisely into the Lebesgue decomposition of the Fourier dual measure. In Section 7, by extending the concept of mean to weakly almost periodic distributions, we further show that for distributions with measure Fourier transform, the Fourier Bohr coefficients of the distribution correspond to the intensity of the atoms in the pure point part of the measure. We thus give a simple characterisation of null weakly almost periodic tempered distributions with measure Fourier trasnform, and subsequently, in Section 8 use this to show that if a weakly almost periodic distribution has a diffraction, it is necessarily pure point. In the final section, we apply the results of the paper in the consideration of several examples, and describe several classes of tempered distributions for which an autocorrelation exists.

Throughout the paper we use $\cS$ to denote the space of Schwartz functions on $\RR^d$ and $\cSp$ for its dual, the space of tempered distributions on $\RR^d$. When we wish to emphasise a particular $\RR^d$, we write $\cS(\RR^d)$. By $\Cc$, we mean the space of all continuous functions of compact support, and $\cD = \cS\cap\Cc$ is the space of infinitely differentiable functions of compact support on $\RR^d$. The space of all bounded, uniformly continuous functions on $\RR^d$ is denoted by $\Cu$.

Recall that for $\psi\in\cSp$, $f\in\cS$, the function $\psi\ast f \in C^{\infty}(\RR^d)$ is defined by
\[ \psi\ast f(t) = \psi(T_t f_{\n}) \,,\]
where $f_{\n}(s) := f(-s)$, $s\in\RR^d$. We use $\lambda$ to denote Lesbegue measure on $\RR^d$.

\section{Translation Boundedness for Tempered Distributions}

In this section we introduce and study the notion of translation boundedness for tempered distributions, based on the existing concept for measures. Tempered distributions, with Schwartz functions as test functions, are the natural objects to work with here, as for many applications we need to work with Fourier transforms. In any case, most of the results in this section can be extended to distributions by working with functions in $\cD$. We plan to study this in a future project.

In the case of measures, there are two equivalent definitions for translation boundedness: a measure $\mu$ on $\RR^d$ is said to be translation bounded if the function $\mu\ast f$ is bounded for each $f\in \Cc$, if and only if the set $\{|\mu|(K+x):x\in\RR^d\}$ is bounded for each compact $K\subseteq\RR^d$. We shall use a version of the first condition to define translation boundedness for tempered distributions. Although the latter condition does not have a direct analogue in this setting, we will show that as for measures, there are several equivalent conditions that may be used to test for translation boundedness.

Let us start by recalling that the topology on the space of test functions $\cS$ is given by the family of norms $\norm{\cdot}_{M,N}$, defined for $M,N \in \ZZ_+$ and $f\in\cS$ by
\[
\| f \|_{M,N}:= \sup_{|\alpha| \leq M, |\beta| \leq N} \sup_{x\in\RR^d}| x^\alpha D^\beta f(x) |\,.
\]
We may sometimes write $\|x^\alpha D^\beta f \|_\infty$ as shorthand for $\sup_{x\in\RR^d}| x^\alpha D^\beta f(x) |$. The topology on $\cSp$ is the weak-$\ast$ topology, that is, the topology induced by the linear functionals $\{ L_f : f\in\cS\}$ on $\cSp$, where for $\psi\in\cSp$, $L_f(\psi):= \psi(f)$. We are now ready to define the notion of translation boundedness for a tempered distribution.

\begin{definition}
A tempered distribution $\psi$ is said to be {\em translation bounded on  $\cS$} if for all $f \in \cS$ the convolution $\psi\ast f$ is a bounded function.
\end{definition}

The next result shows that, as for measures, a tempered distribution is translation bounded if and only if its convolution with every test function is uniformly continuous and bounded.

\begin{proposition}\label{tb Prop 1}
For a tempered distribution $\psi\in \cSp$, the following are equivalent.
\begin{enumerate}
\item[(i)] $\psi$ is translation bounded on $\cS$.
\item[(ii)] For all $f\in\cS$ we have $\psi\ast f \in \Cu$.
\item[(iii)] There exist $M, N \in\ZZ_+$ and $C>0$ such that for all $f \in \cS$
\[
\| \psi *f \|_\infty \leq C \| f \|_{M, N}\,.
\]
\item[(iv)]  The functionals $\{ T_t \psi : t\in\RR^d\}$ are equi-continuous on $\cS$.

\end{enumerate}
\end{proposition}
\begin{proof}
(i)$\implies$(ii): Let $\psi$ be translation bounded and take $f\in\cS$. Recall that for any multi-index $\alpha$, we have
\[ D^{\alpha}(\psi\ast f) = \psi \ast D^{\alpha}f\,.
\]
As for each $f \in \cS$ and multi-index $\alpha$ we have $D^\alpha f \in \cS$, translation boundedness provides for each $f$ and $\alpha$ a constant $C=C_{f,\alpha}$ such that
\[ \|D^{\alpha}(\psi\ast f)\|_{\infty} \leq C \,.  \]
This implies that all the partial derivatives of $\psi*f$ exist and are bounded, and thus that $\psi*f$ is uniformly continuous.

(ii)$\implies$(i) and (iii)$\implies$(i) are trivial.

(ii)$\implies$(iii): We follow here the argument of \cite[Thm.~1.1]{ARMA1}.
By (ii), the mapping $F(f)=\psi*f$ is a linear transformation from $\cS$ into $\Cu$, which (as is easy to check) has closed graph. Therefore it is continuous. Thus there exist $M,N$ and $C$ such that for all $f \in \cS$,
\[
\| \psi *f \|_\infty \leq C \| f \|_{M, N} \,.
\]

(iii)$\Leftrightarrow$(iv): Using the standard boundedness condition equivalent to equi-continuity for functionals in locally convex topological vector spaces, $\{ T_t \psi : t\in\RR^d \}$ is equi-continuous if and only if there exist some $M,N$ and $C$ such that for all $t \in \RR^d$ and all $f \in \cS$ we have
\[
\left|T_t(f) \right| \leq C \| f \|_{M,N} \,.
\]
Replacing $f$ by $f_-$, this becomes exactly the equivalence (iii)$\Leftrightarrow$(iv).
\end{proof}

Proposition \ref{tb Prop 1} (iii) implies that for each $\psi\in\cSp_\infty$, there exist constants $M,N$ and $C$ such that for all $f\in\cS$, $f\neq 0$,
\[
\frac{ \| \psi * f \|_\infty}{\| f \|_{M,N} } \leq C\,.
\]
This suggests the following definition of a norm for translation bounded tempered distributions.

\begin{definition}\label{def tb norm}
For $M,N \in \ZZ_+$, we define $\norm{\cdot}_{M,N}: \cSp\to\RR_+$ by
\[
\| \psi \|_{M,N} := \sup_{ f \in \cS, f \neq 0} \frac{ \| \psi * f \|_\infty}{\| f \|_{M,N} } \,, \;\; \psi\in\cSp
\]
and
\[
\cSp_{M,N} := \{ \psi \in \cSp : \| \psi \|_{M,N} < \infty \} \,.
\]
When $M=N$ we shall simply write $\norm{\cdot}_{N}$ and $\cSp_{N}$ instead of $\norm{\cdot}_{N,N}$ and $\cSp_{N,N}$.
\end{definition}

The following lemma summarises the properties of $\| \cdot \|_{M,N}$ and $\cSp_{M,N}$.

\begin{lemma}\label{lem cSpM,N}
\begin{enumerate}
\item[(i)] If $M_1 < M_2$ and $N_1 < N_2$ then
\[
 \| \psi \|_{M_2,N_2}  \leq  \| \psi \|_{M_1,N_1} \; \mbox{ for all } \, \psi \in \cSp_{M_1,N_1}\,,
\]
and
\[
\cSp_{M_1,N_1} \subseteq \cSp_{M_2,N_2} \,.
\]
\item[(ii)] If $M,N \in \ZZ_+$ and $P=\min\{M,N\}, Q=\max\{M,N\}$ then
\[
\cSp_{P} \subseteq\cSp_{M,N} \subseteq \cSp_{Q}\,.
\]
\item[(iii)]
\[
\cSp_\infty  =\bigcup_{M,N \in \ZZ_+} \cSp_{M,N}=\bigcup_{N \in \ZZ_+} \cSp_{N} \
\]
\item[(iv)] For any $M,N \in \ZZ_+$, $\| \cdot \|_{M,N}$ is a norm on $\cSp_{M,N}$.
\item[(v)] For any $M,N \in \ZZ_+$, the norm  $ \| \cdot \|_{M,N}$ defines a topology stronger than the weak-$\ast$ topology on $\cSp_{M,N}$.
\item[(vi)] For any $M,N \in \ZZ_+$, if $\{\psi_n\}$ is a Cauchy sequence in $(\cSp_{M,N}, \| \cdot \|_{M,N})$ then there exists some $\psi \in \cSp_{M,N}$ such that for all $f \in \cS$ we have
    \[
    \psi_n*f \to \psi*f \, \mbox{ in } (\Cu, \| \cdot \|_\infty)
    \]
\end{enumerate}
\end{lemma}

\begin{proof}
(i) follows from Definition \ref{def tb norm} and the observation that for $f \in \cS$ we have
\[
 \| f \|_{M_1,N_1}  = \sup_{|\alpha| \leq M_1, |\beta| \leq N_1} \| x^\alpha D^\beta f \|_\infty  \leq \sup_{|\alpha| \leq M_2, |\beta| \leq N_2} \| x^\alpha D^\beta f \|_\infty  =  \| \psi \|_{M_2,N_2} \,.
\]

(ii) is an immediate consequence of (i).

(iii): The equality
\[
\cSp_\infty =\bigcup_{M,N \in \ZZ_+} \cSp_{M,N}\,
\]
follows from Proposition \ref{tb Prop 1}. The second equality follows from (ii).

(iv): It is trivial to check that $\| \cdot \|_{M,N}$ is a norm.

(v): This follows immediately from
\[
\left| \psi (f) \right| \leq \| \psi * f_- \|_\infty \leq \| \psi \|_{M,N} \| f \|_{M,N} \,.
\]

(vi): Let $\{\psi_n\}$ be a Cauchy sequence in $(\cSp_{M,N}, \| \cdot \|_{M,N})$. Then there exists a constant $C$ such that $\| \psi_n \|_{M,N} \leq C$ for all $n$.
By (v), $\{\psi_n\}$ is also a vague Cauchy sequence in $\cSp$ and thus converges vaguely to some $\psi \in \cSp$.

Let $f\in\cS$. We show that the sequence $\{\psi_n*f\}$ converges uniformly to $\psi*f$.
Firstly, as $\psi_n \to \psi$ in the weak-$*$ topology, we see that $\{\psi_n*f\}$ converges pointwise to $\psi*f$.

Moreover, for all $m,n$ we have
\[
\| (\psi_n-\psi_m)\ast f \|_\infty \leq \| \psi_n-\psi_m \|_{M,N} \| f\|_{M,N} \,.
\]
Then $\{\psi_n*f\}$ is a Cauchy sequence in $(\Cu, \| \cdot \|_\infty)$ and must converge uniformly to some $g \in \Cu$.
As $\psi_n*f \to g$ uniformly it follows that  $\psi_n*f \to g$ pointwise.

We have that $\psi_n*f$ converges pointwisely to both $g$ and $\psi*f$. Thus
\[
\psi*f = g \in \Cu \,,
\]
that is, $\psi_n*f$ converges uniformly to $\psi*f \in \Cu$.

Finally, as for all $n$,
\[
\| \psi_n *f\|_\infty \leq C \| f\|_{M,N} \,,
\]
the limit $\psi*f$ must also satisfy this inequality: we have
\[
\| \psi *f \|_\infty \leq C \| f\|_{M,N} \,,
\]
and see that $\psi \in \cSp_{M,N}$.
\end{proof}

From the lemma, we see that it is sufficient to consider the case $M=N$. Accordingly, for the remainder of the paper we shall work with the spaces $\cSp_{N}$ and norms $\norm{\cdot}_{N}$.

It is easy to find examples of translation bounded tempered distibutions. For example, as with measures, all compactly supported tempered distributions are translation bounded. This is a direct consequence of the following lemma.

\begin{lemma}\label{conv compact supp is Schwartz}
Let $f \in\cS$ and $\phi\in\cSp$ be a distribution with compact support. Then $\phi*f \in\cS$.
\end{lemma}
\begin{proof}
We know that $\phi\ast f \in \cSp \cap C^\infty(\RR^d)$ \cite[Thm.~4.1.1 and Thm.~7.1.15]{Hor}, and that
\[
\widehat{\phi\ast f}=\widehat{\phi}\widehat{f}  \,.
\]

Also, we have $\widehat{f} \in \cS$ and $\widehat{\phi} \in C^\infty(\RR^d)$ \cite[Thm.~7.1.14]{Hor}.

Now, since $\widehat{\phi} \in C^\infty(\RR^d)$ is tempered as a distribution and $\widehat{f} \in \cS$, it follows immediately that $\widehat{\phi}\widehat{f} \in \cS$ (see for example the proof of \cite[Thm.~7.1.14]{Hor}). Therefore, its inverse Fourier Transform, $\phi\ast f$, is also in $\cS$.
\end{proof}

\begin{proposition}\label{lem_compact_supp}
Let $\phi\in\cSp$ have compact support. Then $\phi$ is translation bounded on $\cS$.
\end{proposition}
\begin{proof}
Let $f\in\cS$. By Lemma \ref{conv compact supp is Schwartz}, $\phi\ast f\in\cS$ and is therefore bounded.
\end{proof}

Applying Lemma \ref{conv compact supp is Schwartz} again, we see that the class of translation bounded tempered distributions is closed under convolution with compactly supported distributions.

\begin{proposition}\label{conv is also tb}
Let $\psi\in\cSp_{\infty}$ and $\phi$ be a distribution with compact support. Then $\phi*\psi\in\cSp_{\infty}$.
\end{proposition}
\begin{proof}
Let $f\in\cS$ be arbitrary. Then, as $\phi$ has compact support, we have
\[
(\phi*\psi)*f= \psi * (\phi *f) \,,
\]
and by Lemma \ref{conv compact supp is Schwartz} $\phi*f \in \cS$. Therefore, as $\psi$ is translation bounded, the function $(\phi*\psi)*f =\psi * (\phi *f)$ is bounded. This completes the proof.
\end{proof}

The class of translation bounded tempered distributions is also closed under taking derivatives.

\begin{proposition}\label{prop Dphi is also tb}
Let $\psi\in\cSp_{\infty}$. Then for any multi-index $\alpha$, $D^{\alpha}\psi\in\cSp_{\infty}$.
\end{proposition}
\begin{proof}
Recall that for a multi-index $\alpha$, $D^{\alpha}\psi$ is the tempered distribution defined for $f\in\cS$ by
$D^{\alpha}\psi(f):= (-1)^{|\alpha|}\psi(D^\alpha f)$.
Then for any $f\in\cS$, we have
\[ \norm{D^{\alpha}\psi \ast f}_{\infty} = \norm{\psi \ast D^{\alpha}f}_{\infty}\,,\]
which is finite as $D^{\alpha}f\in\cS$ and $\psi$ is translation bounded.
\end{proof}

We now examine the relationship between the two types of translation boundedness available to a tempered measure and see that, somewhat counter-intuitively, it is harder for a measure to be translation bounded as a measure than as a tempered distribution. It is well known that a translation bounded measure is tempered; we further show that such a measure must be translation bounded as a tempered distribution. However, there exist tempered measures that are translation bounded as tempered distributions but not as measures.

\begin{proposition}\label{tb measures}
For a measure $\mu$ on $\RR^d$, the following are equivalent.
\begin{itemize}
\item[(i)] $\mu$ is a translation bounded measure.
\item[(ii)] $|\mu|$ is a translation bounded measure.
\item[(iii)] $|\mu|$ is a tempered distribution, which is translation bounded on $\cS$.
\end{itemize}
\end{proposition}
\begin{proof}

(ii) $\Rightarrow$ (i) is obvious and (i) $\Rightarrow$ (ii) is standard; see for example \cite[Proposition 1.12]{bergforst}.

(iii) $\Rightarrow$ (ii): For each compact $K \subseteq \RR^d$ we pick some $f \in \cS$ with $f_- \geq 1_K$. Then, for all $t \in \RR^d$ we have
\[
\left| \mu \right|(t+K) \leq \left| \mu \right| (T_t f_-) \leq \| |\mu|*f \|_\infty \,.
\]
By assumption, $|\mu|*f $ is bounded, so that $\{ \left| \mu \right|(t+K) : t\in\RR^d\}$ is bounded and thus $\mu$ is a translation bounded measure.

(ii) $\Rightarrow$ (iii): Let $f \in \cS$. Firstly, any translation bounded measure is a tempered distribution \cite[Thm.~7.1]{ARMA1}.
Then for all $t \in \RR^d$ we have
\begin{align*}
\left| \left| \mu \right| * f(t) \right| &= \left|\int_{\RR^d}  f(s) d \left| \mu \right| (t-s) \right| \\
  &\leq \sum_{n \in \ZZ^d} \int_{n+[-\frac{1}{2},\frac{1}{2}]^d} \left|  f(s) \right| d \left| \mu \right| (t-s) \\
  &\leq \sum_{n \in \ZZ^d} \int_{n+[-\frac{1}{2},\frac{1}{2}]^d} \sup_{s \in n+[-\frac{1}{2},\frac{1}{2}]^d} \{ \left|  f(s) \right|\}  d \left| \mu \right| (t-s) \\
  &\leq \sum_{n \in \ZZ^d}  \sup_{s \in n+[-\frac{1}{2},\frac{1}{2}]^d} \{ \left|  f(s) \right|\}   \left| \mu \right| (t-n- [-\frac{1}{2},\frac{1}{2}]^d) \\
  &\leq \left( \sum_{n \in \ZZ^d}  \sup_{s \in n+[-\frac{1}{2},\frac{1}{2}]^d} \{ \left|  f(s) \right|\} \right)   \| \mu \|_{[-\frac{1}{2},\frac{1}{2}]^d} \,. \\
\end{align*}
Now, since $f$ is rapidly decaying,
\[
\sum_{n \in \ZZ^d}  \sup_{s \in n+[-\frac{1}{2},\frac{1}{2}]^d} \{ \left|  f(s) \right|\}  < \infty \,.
\]
Therefore
\[
\| \left| \mu \right| * f \|_\infty \leq \left( \sum_{n \in \ZZ^d}  \sup_{s \in n+[-\frac{1}{2},\frac{1}{2}]^d} \{ \left|  f(s) \right|\} \right) \| \mu \|_{[-\frac{1}{2},\frac{1}{2}]^d}  < \infty \,.
\]
\end{proof}

\begin{corollary}\label{TB measures cor} If $\mu$ is a translation bounded measure then $\mu$ is a tempered distribution which is translation bounded on $\cS$.
\end{corollary}
\begin{proof} Since $\mu$ is translation bounded, it is tempered as a distribution. Also for all $f \in \cS$ we have
\[
\left|\mu\ast f(t) \right| \leq  \left| \mu \right|\ast |f| (t) \,.
\]
To complete the proof, we repeat the computation of Proposition~\ref{tb measures}, (ii) $\Rightarrow$ (iii) with $f$ replaced by $|f|$, and obtain
\[
\left| \mu \right| * \left|f\right|(t)
	      \leq \left( \sum_{n \in \ZZ^d}  \sup_{s \in n+[-\frac{1}{2},\frac{1}{2}]^d} \{ \left|  f(s) \right|\} \right)   \| \mu \|_{[-\frac{1}{2},\frac{1}{2}]^d} \,. \\
\]

Exactly as in the proof of Proposition~\ref{tb measures}, (ii) $\Rightarrow$ (iii), this implies that
\[
\|\mu\ast f\|_\infty \leq \infty  \,.
\]
\end{proof}

\begin{remark}\label{tb as td not meas}
The converse of Corollary \ref{TB measures cor} is not true. In \cite[Prop.~7.1]{ARMA1}, the authors introduce a positive definite measure $\mu$ on $\RR$ which which is a tempered distribution, but is not translation bounded as measure.
Corollary \ref{cor posdef is uc} below will show that this measure $\mu$ is, however, translation bounded as a tempered distribution.
\end{remark}

Directly from Corollary \ref{TB measures cor} and Proposition \ref{prop Dphi is also tb}, we have the following.

\begin{corollary}\label{cor Dtbmeasure is tb}
If $\mu$ is a translation bounded measure on $\RR^d$, then for any multi-index $\alpha$, $D^\alpha\mu$ is a translation bounded tempered distribution.
\end{corollary}

\begin{example}\label{eg DdeltaZ}
One can easily compute directly that $\psi\in\cSp(\RR)$ defined by
\[ \psi(f) := \sum_{k\in\ZZ}f'(k)\,,\;\; f\in\cS(\RR)\,,\]
is a translation bounded tempered distribution. Alternatively, by observing that $\psi = -D\delta_{\ZZ}$, we may simply apply the above result.
\end{example}

We shall present some further examples of translation bounded tempered distributions in the subsequent sections.

\section{Autocorrelation}

As with translation boundedness, our definition of autocorrelation for a tempered distribution will be based on that for measures. Recall that the autocorrelation of a measure $\mu$ on $\RR^d$ is defined as the volume averaged convolution $\mu\econv\widetilde{\mu}$, where for $g\in\Cc$,
\[ (\mu\econv\widetilde{\mu})(g) := \lim_{R\to\infty} \tfrac{1}{\vol(B_R)} (\mu_R \ast \widetilde{\mu}_R )(g)\,.\]
Here, $\mu_R$ is the restriction of $\mu$ to the ball $B_R := B_R(0)$, $\widetilde{\mu}(g) := \overline{\mu(\widetilde{g})}$ and $\widetilde{g}:= \overline{g}_{\n}$. The restriction of a measure to a set makes sense due to regularity; that is, regularity ensures that $\mu_R(g) = \mu(1_{B_R} g)$ is well-defined for any $g\in\Cc$. For an arbitrary tempered distribution $\psi$ and $g\in\cS$, however, $\psi(1_{B_R}g)$ is not defined, so to generalise the definition of autocorrelation, we need a different approach to restriction.

\begin{definition}
We say that $\{ h_R \}_{R\geq 1}$ is a \emph{smooth approximate van Hove family} if for each $R$,
\begin{itemize}
\item[(i)] $h_R \in \cD$ and
\item[(ii)] $1_{B_R} \leq h_R \leq 1_{\overline{B_{R+1}}}$.
\end{itemize}
\end{definition}

Given $\psi\in\cSp$ and such a family $\{h_R\}$, each tempered distribution $h_R\psi$, defined for $f\in\cS$ by $h_R\psi(f):= \psi(h_R f)$, has support contained in $B_{R+1}$ and agrees with $\psi$ on the set $\{f\in\cS : \supp(f)\subseteq B_R\}$. Using this, we can define the autocorrelation of a tempered distribution just as for a measure.

\begin{definition}
Given a tempered distribution $\psi$, we say that $\phi$ is an {\em autocorrelation} of $\psi$ if there exists a smooth approximate van Hove family $\{h_R\}$ and a positive sequence $\{R_n\}\to\infty$ such that
\begin{equation}\label{def autocor}
\phi = \lim_{n\to\infty} \frac{1}{\vol(B_{R_n})} \left( h_{R_n} \psi\right)* ( \widetilde{h_{R_n} \psi})   \,
\end{equation}
in the weak-$\ast$ topology in $\cSp$.
\end{definition}

It is not clear to us if every translation bounded tempered distribution possesses an autocorrelation. We can show that any translation bounded tempered distribution would have an autocorrelation in the weak-* topology of $\cD'$, but then the limit would be a distribution that is not necessarily tempered, and hence we could not use the Fourier transform. We will show later that every translation bounded measure possesses autocorrelations in the tempered distribution sense, and that the autocorrelations as tempered distribution and translation bounded measures coincide. Anyhow, for tempered distributions in general the existence of an autocorrelation is still an open problem.

We shall return to such considerations presently. In general we shall assume the existence of an autocorrelation.

As with measures, a tempered distribution may have many autocorrelations, each being a cluster point when $R \to \infty$ of the set
\[ \left\{ \tfrac{1}{\vol(B_R)} \left( h_R \psi\right)* ( \widetilde{h_R \psi}) : R\geq 1\right\} \]
in $\cSp$, for some smooth van Hove family $\{h_R\}$.

\begin{lemma}\label{ac implies precompactness} Let $\psi\in\cSp$ and let $\{h_R\}$ be a  smooth approximate van Hove family. Then $\psi$ has an autocorrelation with respect to $\{ h_R \}$ if and only if there exists some sequence $\{R_n\} \to \infty$ such that the set
\[
\left\{ \tfrac{1}{\vol(B_{R_n})} \left( h_{R_n} \psi\right)* ( \widetilde{h_{R_n} \psi}) : n\in\NN \right\}
\]
is pre-compact in the weak-$\ast$ topology.
\end{lemma}
\begin{proof}
The proof is straightforward, and we will skip it.
\end{proof}

This lemma suggests the following definition.

\begin{definition}
Given a tempered distribution $\psi$ and a smooth approximate van Hove family $\{ h_R \}$, we say that $\{ h_R \}$ is a {\em strong smooth approximate van Hove family} for $\psi$ if the set
\[
\left\{ \tfrac{1}{\vol(B_R)} \left( h_R \psi\right)* ( \widetilde{h_R \psi})  : R \geq 1 \right\}
\]
is weak-$\ast$ pre-compact in $\cSp$.
\end{definition}

If a tempered distribution admits a strong smooth approximate van Hove family, then it has an autocorrelation.
We will show that the autocorrelation(s) of a tempered distribution are independent of the choice of strong smooth van Hove family.
The following lemma allows us to work with functions of the form $f\ast g$, $f,g\in \cD$.

\begin{lemma}\label{conv on cD implies cS}
Let $X \subseteq \cSp$ be weak-$\ast$ compact, $\{\psi_n\} \subseteq X$ and $\psi \in \cSp$ be such that
\[
\psi_n (f) \to \psi(f)
\]
for all $f$ in  a dense set $D \subseteq \cD$. Then $\psi \in X$ and for all $f \in \cS$, we have
\[
\psi_n (f) \to \psi(f) \,.
\]
\end{lemma}
\begin{proof} We first prove that $\psi_n (g) \to \psi(g)$ for all $g \in \cS$.

Indeed, assume by contradiction that for some $g\in\cS$, $\{\psi_n(g)\}$ does not converge to $\psi(g)$. Then we can find an open set $U \ni \psi(g)$, and some increasing sequence $\{k_n\}$ such that $\{\psi_{k_n}(g)\} \nsubseteq U$. Since $\{\psi_{k_n}\}\subseteq X$, which is weak-$\ast$ compact, we can find a subsequence $\{\psi_{l_n}\}$ that converges weak-$\ast$ to some $\psi'\in X$. In particular, $\{\psi_{l_n}(g)\}$ converges to $\psi'(g)$.
As $\{\psi_{l_n}(g)\} \nsubseteq U$ and $U$ is open, it follows that $\psi'(g) \notin U$, so that $\psi'(g) \neq \psi(g)$.

Now for all $f \in D$ we have
\[
\psi(f) = \lim_n \psi_n (f) =\lim_n \psi_{l_n} (f)=\psi'(f) \,.
\]
But $D$ is dense in $\cD$ and hence also in $\cS$, so as $\psi, \psi' \in \cSp$, by continuity we get $\psi=\psi'$, which contradicts $\psi'(g) \neq \psi(g)$.

This shows that $\psi_n \to \psi$ in the weak-* topology of $\cSp$. By compactness of $X$ we also get $\psi \in X$.
\end{proof}

\begin{lemma}\label{L4.5}
Let $\psi\in\cSp$ and $\{h_R\}, \{h'_R\}$ be smooth approximate van Hove families. Then
\[
\lim_{R \to \infty} \left( \tfrac{1}{\vol(B_R)} (h_R\psi)*\widetilde{(h_R\psi)}- \tfrac{1}{\vol(B_R)} (h'_R\psi)*\widetilde{(h'_R\psi)} \right) (f*g) =0\,.
\]
for all $f,g \in \cD$.
\end{lemma}
\begin{proof}

Let $f,g \in \cD$ and choose $R_0>1$ large enough such that the supports of both $f$ and $g$ lie within $B_{R_0}$.
For $R>R_0$, we have $\supp(T_t f) \cap B_{R+1} = \emptyset$ for all $t\not\in B_{R_0+R+1}$ and $\supp(T_t f) \subseteq B_{R}$ for all
$t\in B_{R-R_0}$. Hence
\[
\left(h_R\psi -h'_R\psi \right)*f(t) =0 \quad \mbox{ for } \quad t \notin B_{R+R_0+1} \backslash B_{R-R_0} \,,
\]
and, similarly,
\[
\left(\widetilde{h_R\psi} -\widetilde{h'_R\psi} \right)*g(t) =0 \quad \mbox{ for } \quad t \notin B_{R+R_0+1} \backslash B_{R-R_0} \,.
\]
Now, these functions are continuous and hence bounded on the compact set \\ $\overline{B_{R+R_0+1}} \backslash B_{R-R_0}$, and we have
\begin{eqnarray*}
\begin{split}
&\left( (h_R \psi \ast \widetilde{h_R\psi})-(h'_R\psi\ast\widetilde{h'_R\psi}) \right) \ast (f \ast g)\\
 &= \left((h_R\psi\ast f)\ast (\widetilde{h_R\psi}-\widetilde{h'_R\psi})\ast g\right)+ \left((h_R\psi-h'_R\psi)\ast f \ast (\widetilde{h'_R\psi}\ast g)\right)\,.
\end{split}
\end{eqnarray*}
Therefore, by the van Hove property of $B_R$, we have for all $t\in\RR^d$ that
\[
\lim_R \tfrac{1}{\vol(B_R)} (h_R\psi)*\widetilde{(h_R\psi)}\ast f*g(t)- \tfrac{1}{\vol(B_R)} (h'_R\psi)*\widetilde{(h'_R\psi)}*f*g(t) =0\,.
\]
In particular (taking $t=0$), we have
\[
\lim_R \left(\tfrac{1}{\vol(B_R)} (h_R\psi)*\widetilde{(h_R\psi)}-\tfrac{1}{\vol(B_R)} (h'_R\psi)*\widetilde{(h'_R\psi)}\right)(f*g)=0 \,.
\]
\end{proof}

\begin{lemma}\label{strong implies non-strong} Let $\{h_R\}, \{h'_R\}$ be smooth approximate van Hove families such that $h_R$ is strong for $ \psi \in \cSp_\infty$. Then any autocorrelation $\phi$ of $\psi$ that is calculated with respect to $h'_R$ is also an autocorrelation with respect to $h_R$, and is given by the same choice of sequence $R_n \to \infty$.
\end{lemma}
\begin{proof}
We need to show that
\[
 \lim_n \tfrac{1}{\vol(B_{R_n})} \left( h'_{R_n} \psi\right)* ( \widetilde{h'_{R_n} \psi}) =\phi \;\; \mbox{ in } \cSp
\]
implies
\[
 \lim_n \tfrac{1}{\vol(B_{R_n})} \left( h_{R_n} \psi\right)* ( \widetilde{h_{R_n} \psi}) = \phi \;\; \mbox{ in } \cSp\,.
\]
As
\[
X:= \overline{ \left\{ \tfrac{1}{\vol(B_R)} \left( h_R \psi\right)* ( \widetilde{h_R \psi})  : R \geq 1 \right\} }
\]
is weak-$\ast$ compact in $\cSp$, and by the existence of the limit the set
\[
\{ \tfrac{1}{\vol(B_{R_n})} \left( h'_{R_n} \psi\right)* ( \widetilde{h'_{R_n} \psi}) | n \}
\]
is also pre-compact, the claim follows immediately by combining Lemmas \ref{conv on cD implies cS} and \ref{L4.5}.
\end{proof}

\begin{proposition}\label{smooth van Hove is irrelevant} Let $\{h_R\}, \{h'_R\}$ be strong smooth approximate van Hove families for $\psi$.
If $\{R_n\} \to \infty$, then
\[
 \lim_n \frac{1}{\vol(B_{R_n})} \left( h_{R_n} \psi\right)* ( \widetilde{h_{R_n} \psi})  \mbox{ exists }
\]
if and only if
\[
 \lim_n \frac{1}{\vol(B_{R_n})} \left( h'_{R_n} \psi\right)* ( \widetilde{h'_{R_n} \psi})  \mbox{ exists } \,,
\]
and in this case they are the same.
\end{proposition}
\begin{proof}
The claim follows immediately by applying Lemma \ref{strong implies non-strong} twice.
\end{proof}

The previous result shows us that whenever when we deal with strong smooth approximate van Hove families, the autocorrelation is independent of this choice.

If a tempered distribution does not admit an autocorrelation, it cannot have a strong smooth approximate van Hove sequence. It is presently unclear to us if the converse is true, and we do not have a nice characterisation of the distributions that admit strong smooth approximate van Hove families. However, we will show that for translation bounded measures, all smooth approximate van Hove families are strong, and the autocorrelations that they give rise to correspond to those arising from the standard definition for measures.

For the rest of the paper all the smooth approximate van Hove families will be assumed to be strong. This allows us to use the following convention: whenever $\psi \in \cSp_\infty$ and $\{h_R\}$ is a smooth van Hove family, we define
\[
\psi_R:=h_R \psi
\]
for simplicity.

\begin{proposition}\label{tb measures have strong}
Let $\mu \in \cM^\infty(\RR^d)$ and $\{h_R\}$ be a smooth approximate van Hove family. Then $\{h_R\}$ is strong for $\mu$.
\end{proposition}
\begin{proof}
Let $f \in \cS$. Without loss of generality we can assume that $f \geq 0$, otherwise we can replace $f$ by $g\in\cS$ such that $g\geq|f|$.

As $\mu$ is translation bounded, $|\mu|$ is translation bounded as a tempered distribution by Proposition \ref{tb measures} and thus there exist some $M,N$ and $C$ such that
\[
\| | \mu | \ast f \|_\infty \leq C \| f \|_{M,N} \,.
\]
Then we have
\begin{eqnarray*}
\begin{split}
\| \tfrac{1}{\vol(B_{R})} \left( h_{R} \mu \right)* ( \widetilde{h_{R} \mu})*f  \|_\infty \\
&= \tfrac{1}{\vol(B_{R})} \| (\widetilde{h_{R} \mu})* \bigl( ( h_{R} \mu )*f\bigr)  \|_\infty \\
&= \tfrac{1}{\vol(B_{R})} \sup_{x \in \RR^d}\left|\int_{\RR^d} \bigl(( h_{R}\mu) \ast f \bigr) (x-t) d \widetilde{h_{R} \mu} (t)  \right| \\
&\leq \tfrac{1}{\vol(B_{R})}\sup_{x \in \RR^d}\int_{\RR^d}\left| \left( h_{R} \mu\right) \ast f\right|(x-t)\, d \widetilde{\left| h_{R} \mu \right|}(t) \\
&\leq  \tfrac{1}{\vol(B_{R})} \sup_{x \in \RR^d} \int_{\RR^d}\left| h_{R} \mu \right|\ast f (x-t) d \widetilde{\left| h_{R} \mu \right|} (t)  \\
&\leq  \tfrac{1}{\vol(B_{R})} \sup_{x \in \RR^d} \int_{\RR^d}\left|\mu \right|\ast f (x-t)\, d \widetilde{\left| h_{R} \mu \right|} (t)  \\
&\leq  \tfrac{1}{\vol(B_{R})} \sup_{x \in \RR^d} \int_{\RR^d}  C \| f \|_{M,N}\, d \widetilde{\left| h_{R} \mu \right|} (t)  \\
&=\tfrac{1}{\vol(B_{R})}  C \| f \|_{M,N} \widetilde{\left| h_{R} \mu \right|} (\RR^d)  \\
&\leq  \tfrac{1}{\vol(B_{R})}  C \| f \|_{M,N} \widetilde{ \left|  \mu \right|} (B_{R+1})\,.  \\
\end{split}
\end{eqnarray*}

Now, a simple computation shows that for any translation bounded measure there exists a constant $C_0$ such that for all $R>1$ we have
\[
\frac{1}{\vol(B_{R})}   \left|  \mu \right| (B_{R+1}) \leq C_0 \,.
\]

Therefore, for each $f \in \cS$ we have for all $R>1$ that
\[
\left\| \tfrac{1}{\vol(B_{R})} \left( h_{R} \mu \right)* ( \widetilde{h_{R} \mu})*f  \right\|_\infty \leq C_0  C \| f \|_{M,N}
\]
and thus the set
\[
\left\{ \tfrac{1}{\vol(B_{R})} \left( h_{R} \mu \right)* ( \widetilde{h_{R} \mu}) : R \geq 1 \right\}
\]
is weak-$\ast$ precompact, as required.
\end{proof}

Next, we show that derivatives of translation bounded measures also admit strong smooth van Hove sequences.

\begin{proposition}\label{derivative tb measures}
Let $\{h_R\}$ be a smooth approximate van Hove family and let $\kappa \in \ZZ_+^d$. If there exists a constant $C$ such that for all $\alpha \leq \kappa$ and all $R \geq 1$ we have
\[
\|D^\alpha h_R \|_\infty < C \,,
\]
then for each $\mu\in\cM^{\infty}(\RR^d)$, the family $\{h_R\}$ is strong for the tempered distribution $D^\kappa \mu$.
\end{proposition}
\begin{proof}
In this proof we will use the notation
\[
{\bf 0}:= (0,0,0,...,0) \; \mbox{ and } {\bf 2}:= (2,2,2,...,2)\,  \in \ZZ^d
\]
in order to differentiate between these vectors and the numbers $0,2\in\ZZ$. We begin by noting that since $h_R \equiv 1$ on $B_R$ and $h_R \equiv 0$ outside $\overline{B_{R+1}}$ we have for all non-zero $\alpha \in \ZZ^d_+$
\[
D^\alpha h_R(x) = 0 \, \mbox{ for all } \, x \notin  \overline{B_{R+1}} \backslash B_R \,.
\]
Moreover, for $\alpha = \bf{0} \in \ZZ^d_+$ we have $\|D^\alpha h_R \|_\infty =1$ for all $R$.

Now let $\mu \in \cM^{\infty}(\RR^d)$ and $f \in \cS$.
Since $\widetilde{h_R D^M \mu}$ is a distribution with compact support and $f \in \cS$, by Lemma \ref{conv compact supp is Schwartz} we have
\[
g:= \bigl( \widetilde{h_R D^\kappa \mu} \bigr) *f \in \cS \,.
\]
Then for all $x \in \RR^d$,
\begin{eqnarray*}
\begin{split}
\left| \tfrac{1}{\vol(B_{R})} \left( h_{R} D^\kappa \mu \right)* g (x) \right|&= \tfrac{1}{\vol(B_{R})} \left| \langle h_{R} D^\kappa \mu , T_x g_- \rangle  \right| \\
&= \tfrac{1}{\vol(B_{R})} \left| \langle D^\kappa \mu ,  h_{R} T_x g_- \rangle  \right| \\
&= \tfrac{1}{\vol(B_{R})} \left| \langle \mu , D^\kappa \bigl( h_{R} T_x g_- \bigr) \rangle  \right| \\
&= \tfrac{1}{\vol(B_{R})}  \left|\int_{\RR^d}  D^\kappa \bigl( h_{R} T_x g_- \bigr)(t) d \mu(t)   \right|  \,.
\end{split}
\end{eqnarray*}
Now, by the multivariate Leibnitz product rule we have
\[
D^\kappa \bigl( h_{R} T_x g_- \bigr)(t) = \sum_{ \beta \leq \kappa} \binom{\kappa}{\beta}  \bigl( D^\beta h_{R}  \bigr) \bigl( D^{\kappa- \beta} T_x g_- \bigr)(t) \,.
\]
Since for all $\beta \neq {\bf 0}$ we have $ D^\beta h_{R}  \equiv 0$ outside  $\overline{B_{R+1}} \backslash B_R$, we also have
\begin{eqnarray*}
\begin{split}
\int_{\RR^d}  D^\kappa \bigl( h_{R} T_x g_- \bigr)(t) d \mu(t)  &= \sum_{ {\bf 0} < \beta \leq \kappa} \binom{\kappa}{\beta} \int_{\overline{B_{R+1}} \backslash B_R}  D^\beta \bigl( h_{R} \ T_x g_- \bigr)(t) d \mu(t) \\
&+\int_{\RR^d}   h_{R} \bigl( D^\kappa T_x g_- \bigr)(t) d \mu(t) \,.
\end{split}
\end{eqnarray*}
Therefore we get
\begin{equation}\label{very ugly}
\begin{split}
&\left| \tfrac{1}{\vol(B_{R})} \left( h_{R} D^\kappa \mu \right)* g (x) \right|\leq \tfrac{1}{\vol(B_{R})}\sum_{ {\bf 0} < \beta \leq \kappa}  \binom{\kappa}{\beta} \left| \int_{\overline{B_{R+1}} \backslash B_R}  D^\beta \bigl( h_{R} \ T_x g_- \bigr)(t) d \mu(t) \right|\\
&+ \tfrac{1}{\vol(B_{R})}\left| \int_{\RR^d}   h_{R} \bigl( D^\kappa T_x g_- \bigr)(t) d \mu(t)  \right| \\
&\leq \tfrac{1}{\vol(B_{R})}\sum_{ {\bf 0} < \beta \leq \kappa}  \binom{\kappa}{\beta} \sum_{ \alpha \leq  \beta}  \binom{\beta}{\alpha} \left| \int_{\overline{B_{R+1}} \backslash B_R}  ( D^\alpha h_{R} ) \bigl(D^{\beta-\alpha} T_x g_- \bigr)(t) d \mu(t) \right|\\
&+ \tfrac{1}{\vol(B_{R})}\int_{B_{R+1}}  \left|  D^\kappa T_x g_- (t) \right| d |\mu|(t)  \\
&\leq \tfrac{1}{\vol(B_{R})}\sum_{ {\bf 0} < \beta \leq \kappa} \sum_{ \alpha \leq  \beta} \binom{\kappa}{\beta}   \binom{\beta}{\alpha} C \int_{\overline{B_{R+1}} \backslash B_R} \left| D^{\beta-\alpha} T_x g_- \right|(t) d |\mu|(t) \\
&+ \tfrac{1}{\vol(B_{R})}\int_{B_{R+1}}  \left|  D^\kappa T_x g_- (t) \right| d |\mu|(t)  \\
&\leq \tfrac{C_1}{\vol(B_{R})} \sup_{\gamma\leq\kappa} \int_{\overline{B_{R+1}} \backslash B_R} \left| D^{\gamma} T_x g_-(t) \right| d |\mu|(t) + \tfrac{1}{\vol(B_{R})}\int_{B_{R+1}}  \left|  D^\kappa T_x g_- (t) \right| d |\mu|(t)  \,.
\end{split}
\end{equation}
where
\[
C_1:= C \cdot \bigl( \sum_{ {\bf 0} < \beta \leq \kappa} \sum_{ \alpha \leq  \beta} \binom{\kappa}{\beta}   \binom{\beta}{\alpha} \bigr)
\]
is a constant which only depends on $C$ and $\kappa$.

To complete the proof we will show that there exists a constant $C_4$ which depends only on $\kappa$ and $\mu$, but it is independent of $R$ and $f$, such that for all $\gamma \leq \kappa$ we have
\[
\| D^\gamma  g \|_\infty < C_4  \| f \|_{2d, 2|\kappa|}  \,.
\]
Let us emphasize here that while $g \in \cS$, $g$ depends on $R$ so should be seen as a family of functions, not a single function.

Now, for all $\gamma \leq \kappa$ we have
\begin{eqnarray*}
\begin{split}
\left|  D^\gamma  g (t) \right|&=\left| D^\gamma \bigl( \widetilde{h_R D^\kappa \mu} \bigr) *f \right|(t)  \\
&=\left| \bigl( \widetilde{h_R D^\kappa \mu} \bigr) *D^\gamma f \right|(t)  \\
&=\left| \langle \widetilde{h_R D^\kappa \mu} , T_t (D^\gamma f)_- \rangle \right| \\
&=\left| \langle h_R D^\kappa \mu , T_{-t} \overline{(D^\gamma f)}  \rangle \right| \\
&=\left| \langle  D^\kappa \mu , h_R T_{-t} \overline{(D^\gamma f)} \rangle \right| \\
&=\left| \langle  \mu , D^\kappa h_R T_{-t} \overline{(D^\gamma f)}  \rangle \right| \\
&=\left| \langle  \mu ,  \sum_{ \beta \leq \kappa} \binom{\kappa}{\beta} (D^\beta h_R) T_{-t} \overline{(D^{\kappa-\beta+\gamma} f)}  \rangle \right| \\
&\leq \sum_{ \beta \leq \kappa} \binom{\kappa}{\beta}  \left| \langle  \mu ,  (D^\beta h_R) T_{-t} \overline{(D^{\kappa-\beta+\gamma} f)}  \rangle \right| \\
&=\left| \langle  \mu ,  h_R T_{-t} \overline{(D^{\kappa+\gamma} f)}  \rangle \right| +\sum_{{\bf 0} < \beta \leq \kappa} \binom{\kappa}{\beta}  \left| \langle  \mu ,  (D^\beta h_R) T_{-t} \overline{(D^{\kappa-\beta+\gamma} f)}  \rangle \right| \\
&=\left| \langle  h_R \mu , T_{-t} \overline{(D^{\kappa+\gamma} f)}  \rangle \right| +\sum_{{\bf 0} < \beta \leq \kappa} \binom{\kappa}{\beta}  \left| \langle (D^\beta h_R)  \mu ,  T_{-t} \overline{(D^{\kappa-\beta+\gamma} f)}  \rangle \right| \,.
\end{split}
\end{eqnarray*}

Next, exactly as in the proof of Proposition \ref{tb measures} we get
\begin{eqnarray*}
\begin{split}
\left| \langle  h_R \mu , T_{-t} \overline{(D^{\kappa+\gamma} f)}  \rangle \right| &\leq  \left( \sum_{n \in \ZZ^d}  \sup_{s \in n+[-\frac{1}{2},\frac{1}{2}]^d} \{ \left|  D^{\kappa+\gamma} f (s) \right|\} \right) \| h_R \mu \|_{[-\frac{1}{2},\frac{1}{2}]^d} \\
&\leq \left( \sum_{n \in \ZZ^d}  \sup_{s \in n+[-\frac{1}{2},\frac{1}{2}]^d} \{ \left|  D^{\kappa+\gamma} f (s) \right|\} \right) \| \mu \|_{[-\frac{1}{2},\frac{1}{2}]^d} \,,
\end{split}
\end{eqnarray*}
and similarly, using also that $ | (D^\beta h_R) \mu | \leq C | \mu |$, we get
\begin{eqnarray*}
\begin{split}
\left| \langle (D^\beta h_R)  \mu ,  T_{-t} \overline{(D^{\kappa-\beta+\gamma} f)}  \rangle \right| &\leq \left( \sum_{n \in \ZZ^d}  \sup_{s \in n+[-\frac{1}{2},\frac{1}{2}]^d} \{ \left|  (D^{\kappa-\beta+\gamma} f)  (s) \right|\} \right) \| (D^\beta h_R) \mu \|_{[-\frac{1}{2},\frac{1}{2}]^d}  \\
&\leq \left( \sum_{n \in \ZZ^d}  \sup_{s \in n+[-\frac{1}{2},\frac{1}{2}]^d} \{ \left|  (D^{\kappa-\beta+\gamma} f)  (s) \right|\} \right) C \cdot\| \mu \|_{[-\frac{1}{2},\frac{1}{2}]^d} \,.
\end{split}
\end{eqnarray*}

Finally, using the convergence of the series
\[
\sum_{n \in \ZZ^d \backslash {\bf 0}} \sup_{s \in n+[-\frac{1}{2},\frac{1}{2}]^d} \{ \frac{1}{s^{\bf 2}} \} \,,
\]
we get
\begin{eqnarray*}
\begin{split}
&\sum_{n \in \ZZ^d} \sup_{s \in n+[-\frac{1}{2},\frac{1}{2}]^d} \{ \left|  (D^{\kappa-\beta+\gamma} f)  (s) \right|\} = \sup_{s \in [-\frac{1}{2},\frac{1}{2}]^d} \{ \left|  (D^{\kappa-\beta+\gamma} f)  (s) \right|\} \\
&+ \sum_{n \in \ZZ^d \backslash {\bf 0}} \sup_{s \in n+[-\frac{1}{2},\frac{1}{2}]^d} \{ \left| \frac{1}{s^{\bf 2}} \bigl( s^{\bf 2} (D^{\kappa-\beta+\gamma} f \bigr)  (s) \right|\} \\
 &\leq  \| f \|_{0, |\kappa-\beta+\gamma|} +  \sum_{n \in \ZZ^d \backslash {\bf 0}} \sup_{s \in n+[-\frac{1}{2},\frac{1}{2}]^d} \{ \left| \frac{1}{s^{\bf 2}} \right|   \| f \|_{2d, |\kappa-\beta+\gamma|}  \} \leq C_3  \| f \|_{2d, 2|\kappa|} \,,
\end{split}
\end{eqnarray*}
where
\[
C_3=1+ \sum_{n \in \ZZ^d \backslash {\bf 0}} \sup_{s \in n+[-\frac{1}{2},\frac{1}{2}]^d} \{ \frac{1}{s^{\bf 2}} \} \,.
\]
Combining, we get
\begin{eqnarray*}
\begin{split}
\left|  D^\gamma  g (t) \right|&=\left| D^\gamma \bigl( \widetilde{h_R D^\kappa \mu} \bigr) *f \right|(t)  \\
&\leq\left| \langle  h_R \mu , T_{-t} \overline{(D^{\kappa+\gamma} f)}  \rangle \right| +\sum_{{\bf 0} < \beta \leq \kappa} \binom{\kappa}{\beta}  \left| \langle (D^\beta h_R)  \mu ,  T_{-t} \overline{(D^{\kappa-\beta+\gamma} f)}  \rangle \right| \\
&\leq C_3 \| f \|_{2d, 2|\kappa|}  \| \mu \|_{[-\frac{1}{2},\frac{1}{2}]^d} +\sum_{{\bf 0} < \beta \leq \kappa} \binom{\kappa}{\beta}   C_3 \| f \|_{2d, 2|\kappa|} C  \| \mu \|_{[-\frac{1}{2},\frac{1}{2}]^d} \\
&=C_4  \| f \|_{2d, 2|\kappa|}\,.
\end{split}
\end{eqnarray*}

Thus  by (\ref{very ugly}), we have
\begin{eqnarray*}
\begin{split}
&\left| \tfrac{1}{\vol(B_{R})} \left( h_{R} D^\kappa \mu \right)* \widetilde{h_{R} D^\kappa \mu} * f (x) \right|\\
&\leq \tfrac{C_1}{\vol(B_{R})} C_4  \| f \|_{2d, 2|\kappa|} \left| \mu \right| (\overline{B_{R+1}} \backslash B_R) + \tfrac{1}{\vol(B_{R})} C_4  \| f \|_{2d, 2|\kappa|} \left| \mu \right| (B_{R+1})    \,.
\end{split}
\end{eqnarray*}

By translation boundedness of $\mu$ the sets
\[
\{ \tfrac{1}{\vol(B_{R})}\left| \mu \right| (\overline{B_{R+1}} \backslash B_R) \, | \, R \geq 1 \}
\]
and
\[
\{  \tfrac{1}{\vol(B_{R})}  \left| \mu \right| (B_{R+1}) \,  |  \, R \geq 1 \}
\]
are bounded. Therefore, there exists a constant $C_5$, which depends only on $\kappa$ and $\mu$, such that, for all $R \geq 1$ and all $f \in \cS$ we have
\[
\| \tfrac{1}{\vol(B_{R})} \left( h_{R} D^\kappa \mu \right)* \widetilde{h_{R} D^\kappa \mu} * f \| \leq C_5 \| f \|_{2d, 2|\kappa|} \,.
\]

This completes the proof.
\end{proof}

By the observations at the start of the proof, we see that the boundedness condition in the statement  Proposition  \ref{derivative tb measures} is needed only for the coronas $\overline{B_{R+1}} \backslash B_R$. The existence of such van Hove families is easy to prove. Therefore, all derivatives of translation bounded measures admit strong smooth approximate van Hove sequences.

We complete the section by showing that for a translation bounded measure $\mu$, the autocorrelations calculated for the tempered distribution $\mu$ coincide with those for the measure $\mu$.

\begin{proposition} Let $\mu \in \cM^\infty(\RR^d)$, let $\{h_R\}$ be a smooth approximate van Hove family, and let $\phi$ be an autocorrelation of $\mu$ calculated for $\mu$ as a tempered distribution with respect to the sequence $\{R_n\}$. Then the measure autocorrelation
\[
\gamma= \lim_n  \tfrac{1}{\vol(B_{R_n})} \left( \mu_{R_n} \right)* ( \widetilde{\mu_{R_n}})
\]
exists and
\[
\gamma = \phi \,.
\]
In particular, $\phi$ is a translation bounded measure.
\end{proposition}
\begin{proof}
Exactly as in the proof of Lemma \ref{L4.5}, we can prove that
\begin{equation}\label{lim zero}
\lim_{n \to \infty} \frac{1}{\vol(B_{R_n})} \bigl( \left( h_{R_n} \mu \right)* ( \widetilde{h_{R_n}} \mu )
					    - \left( \mu_{R_n} \right)*(\widetilde{\mu_{R_n}}) \bigr) (f*g) =0\,.
\end{equation}
for all $f,g \in \cD$.
As $\mu$ is a translation bounded measure, the sequence
\[
\left\{ \tfrac{1}{\vol(B_{R_n})} \left( \mu_{R_n} \right)* (\widetilde{\mu_{R_n}}) \right\}
\]
is precompact in the vague topology and thus has a cluster point. We show that the sequence is convergent by showing that the cluster point is unique.

Let $\gamma_1, \gamma_2$ be two cluster points of $\left\{\frac{1}{\vol(B_{R_n})} \left( \mu_{R_n} \right)* (\widetilde{\mu_{R_n}})\right\}$. By using (\ref{lim zero}) on the two subsequences which give $\gamma_1$ and $\gamma_2$, we get
\[
\gamma_1 (f*g)= \phi (f*g) =\gamma_2 (f*g) \,,
\]
for all $f,g \in \cD$. As the set $\{ f*g : f,g \in \cD \}$ is dense in $\Cc$, this shows that $\gamma_1=\gamma_2$.
Thus there exists a measure $\gamma \in \cM^\infty(\RR^d)$ such that
\[
\gamma= \lim_n \tfrac{1}{\vol(B_{R_n})} \left( \mu_{R_n} \right)* (\widetilde{\mu_{R_n}}) \,,
\]
in the vague topology of measures.
Using again (\ref{lim zero}), we get for all $f,g \in \cD$ that
\[
\phi (f*g) =\gamma (f*g) \, .
\]
As $\gamma$ is also a tempered distribution, and $\{ f*g : f,g \in \cD \}$ is dense in $\cS$, we see that $\gamma=\phi$ as tempered distributions. This completes the proof.
\end{proof}

\section{Positive definite tempered distributions}\label{section pos def}

In this section we prove some basic results about positive definite tempered distributions. As any autocorrelation of a tempered distribution is positive definite by construction, the results of this section are important to the study of diffraction of tempered distributions.
We begin with a pretty standard result (see, for example, \cite[Theorem 7.19]{rudinFA}).

\begin{lemma}\label{lem conv in S}
Let $\psi\in\cSp$ and $f\in\cS$. Then the function $\psi\ast f$ has (at most) polynomial growth, that is, $\psi\ast f\in\cSp$, and
$\widehat{\psi\ast f} = \widehat{\psi}\widehat{f}$.
\end{lemma}

\begin{proposition}\label{convolution is CU}
Let $\psi \in \cSp$ such that $\widehat{\psi}$ is a measure and let $f \in\cS$. Then $\psi\ast f$ is bounded and uniformly continuous.
\end{proposition}
\begin{proof}
As $\widehat{f} \in\cS$ and $\widehat{\psi}$ is a tempered measure, $\widehat{\psi}\widehat{f}$ is also a tempered measure. But
\[ \left|\widehat{\psi}\widehat{f}(\RR^d)\right| = \left|\int_{\RR^d} \widehat{f}(x)\dd\widehat{\psi}(x)\right| = \left|\widehat{\psi}(\widehat{f})\right|\,,\]
which is finite, so $\widehat{\psi}\widehat{f}$ is a finite measure and thus, as the Fourier transform of a finite measure, $\psi\ast f$ is uniformly continuous and bounded. In particular, for all $x \in \RR^d$,
\begin{align*}
\left|\psi\ast f(x)\right|  &= \left|(\widehat{\psi}\cdot\widehat{f})^\vee (x) \right|\\
			    &= \left| \int_{\widehat{\RR^d}} \chi(x) \widehat{f}(\chi) \dd\widehat{\psi}(\chi) \right| \\
			    &\leq \int_{\widehat{\RR^d}} \left|\widehat{f}(\chi)\right|\dd|\widehat{\psi}|(\chi)\,,
\end{align*}
so that
\[
\| \psi\ast f \|_\infty \leq \int_{\widehat{\RR^d}} \left| \widehat{f}(\chi) \right|\dd|\widehat{\psi}|(\chi) = |\widehat{\psi}|(|\widehat{f}|) \,.
\]
\end{proof}

\begin{corollary}\label{cor measure FT is uc}
Let $\psi\in\cSp$ such that $\widehat{\psi}$ is a measure. Then $\psi$ is translation bounded on $\cS$.
\end{corollary}

Another way to state this is that the Fourier transform of any tempered measure is a tempered distribution that's translation bounded on $\cS$.
This is another fruitful source of examples of translation bounded tempered distributions. The following example first appeared in \cite{teba}.

\begin{example}\label{eg aperiodic}
Let
\[\omega := \delta_{2\ZZ} + \sum_{n\geq 1} \delta_{2.4^n \ZZ} \ast (\delta_{4^n - 1} + \delta_{1-4^n}) \,.\]
One may easily verify that $\omega$ is a translation bounded measure, and thus tempered. The Fourier transform of $\omega$ has formal expression
\[ \widehat{\omega} = \tfrac{1}{2}\delta_{\frac{\ZZ}{2}} + \sum_{n\geq 1} \frac{\cos(2\pi(4^n-1)(\cdot))}{4^n}\delta_{\frac{\ZZ}{2.4^n}}\,,\]
and is a non-measure tempered distribution. By Corollary \ref{cor measure FT is uc}, it is translation bounded.
\end{example}

\begin{example}
The tempered distribution $\psi = -D\delta_{\ZZ}$ of Example \ref{eg DdeltaZ} has Fourier transform
\[ \widehat{-D\delta_{\ZZ}} = 2\pi \ii\, \id\,\delta_{\ZZ} \,,\]
where $\id$ is the identity function on $\RR$, $\id(x) = x$. That is, $\psi$ is the Fourier transform of a tempered (not translation bounded) measure.
\end{example}

Recall that, from the Bochner Schwartz theorem, the Fourier transform of a positive definite tempered distribution is a tempered measure. Then the following is immediate from Corollary \ref{cor measure FT is uc}.

\begin{corollary}\label{cor posdef is uc}
Let $\psi\in\cSp$ be positive definite. Then $\psi$ is translation bounded on $\cS$.
\end{corollary}

To conclude this section, we further examine the connection between positive definite tempered distributions and those whose Fourier transform is a measure.

\begin{proposition} \label{PD and measure FT} Let $\psi \in \cSp$.
\begin{itemize}
\item[i)] If $\psi$ is a linear combination of 4 positive definite tempered distributions, then $\widehat{\psi}$ is a tempered measure.
\item[ii)] If $\widehat{\psi}$ is a translation bounded measure, then $\psi$  is a linear combination of 4 positive definite tempered distributions.
\end{itemize}
\end{proposition}
\begin{proof}
(i) is clear.

(ii)
Let $f \in \cS$. As $\widehat{\psi}$ is a translation bounded measure, we can take its canonical decomposition
\[
\widehat{\psi}= \left(\widehat{\psi} \right)_{+} - \left(\widehat{\psi} \right)_{-} +i\left( \widehat{\psi} \right)_{i+} -i \left(\widehat{\psi} \right)_{i-} \,.
\]
Then each of the four measures is a positive translation bounded measure, hence a positive tempered measure. Taking the inverse Fourier transform of both sides, we have the result.
\end{proof}

\begin{remark}\label{Tempered non TB} In \cite[Prop.~7.1]{ARMA}, the authors introduce a tempered measure $\mu$ such that the measure $\left| \mu \right|$ is not tempered.

Note that in this case, if we take $\psi = \check{\mu}$, the distribution $\psi$ is tempered and has measure Fourier transform, but it cannot be a linear combination of finitely many positive definite tempered distributions.

Indeed if we would have
\[
\psi = \sum_{i=1}^N c_i \psi_i\,,
\]
with $\psi_i$ positive definite, then
\[
\mu=\widehat{\psi}=\sum_{i=1}^N c_i \widehat{\psi_i}\,,
\]
where each $\widehat{\psi_i}$ is a positive tempered measure. Then
\[
\left| \mu \right| \leq \sum_{i=1}^N |c_i| \widehat{\psi_i}
\]
and as the RHS is tempered so must be the LHS too.
\end{remark}

In summary, we have
\begin{align*}
\{ \psi\in\cSp : \widehat{\psi} \mbox{ is a translation bounded }& \mbox{measure} \} \\
	\subsetneq \mbox{span} \{ \psi\in\cSp :  \psi &\mbox{ is positive definite} \} \\
	&\quad \subsetneq \{ \psi\in\cSp : \widehat{\psi} \mbox{ is a tempered measure} \}\,.
\end{align*}

\section{Almost periodic Distributions}

It is well known \cite{EBE} that the Fourier transform of a finite measure $\mu$ is a weakly almost periodic function $f$ and that, under the Fourier transform, the Lebesgue decomposition of $\mu$ into its pure point and continuous parts corresponds to the Eberlein decomposition of the weakly almost periodic function $f$ \cite{EBE2}.

Based on the above results, Argabright and deLamadrid extended the notion of almost periodicity from functions to measures via convolution, and showed that the correspondence of the respective decompositions under the Fourier transform can be extended to translation bounded Fourier transformable measures \cite{ARMA}.

We shall show that the same approach also works for tempered distributions. In this section we define the Eberlein decomposition for tempered distributions and in the next section, we prove the equivalence of the decompositions under the Fourier transform.
As usual, we denote the spaces of weakly almost periodic, strongly almost periodic and null weakly almost periodic functions by $WAP(\RR^d)$, $SAP(\RR^d)$ and  $WAP_0(\RR^d)$ respectively. For the definitions of these spaces and a review of the existing theory in the setting of measures, we refer the reader to \cite{MoSt}.

\begin{definition}
A tempered distribution $\psi$ is called respectively {\em weakly, strongly} or {\em null weakly almost periodic} if and only if for all $f \in \cS(\RR^d)$ the function $\psi *f$ is  respectively weakly, strongly or null weakly almost periodic.
We shall denote the corresponding spaces of tempered distributions by $\fWAP(\RR^d), \fSAP(\RR^d)$ and $\fWAP_0(\RR^d)$ respectively.
\end{definition}

\begin{remark}
As almost periodic functions must be uniformly continuous and bounded, it is implicitly assumed in the definition that $\psi *f \in \Cu$ for all $f \in \cS(\RR^d)$. Therefore $\fWAP(\RR^d), \fSAP(\RR^d)$ and $\fWAP_0(\RR^d)$ are subspaces of $\cSp_\infty(\RR^d)$.
\end{remark}

\begin{remark}
In \cite{LAG}, Lagarias introduced a notion of uniformly almost periodic distributions which is similar to our definition, but uses $\cD$ as the space of test functions instead of $\cS$. All of the results we prove about translation boundedness and Eberlein convolution can be shown with either $\cD$ or $\cS$ as test functions. We will show next in Lemma \ref{Lemma Lag} that for translation bounded tempered distributions the two definitions are equivalent.
\end{remark}

Let us recall that if $f,g \in \cS$ then $f*g \in \cS$. This allows us look at convolutions of the form
\[
\psi*(f*g) \,;\, \psi \in \cSp , f,g \in \cS\,.
\]

Recall from the previous section that for $\psi\in\cSp, f\in\cS$, we have $\psi\ast f\in\cSp$. Further, we have the following (see for example \cite[Theorem 7.19]{rudinFA}).

\begin{lemma}\label{triple conv}
Let $\psi \in \cSp$ and $f,g \in \cS$. Then $(\psi*f)*g$ and  $\psi*(f*g)$ are well defined and
\[
(\psi*f)*g = \psi*(f*g)\,.
\]
\end{lemma}

\begin{lemma}\label{Lemma Lag} Let $\psi \in \cSp$. Then
\begin{itemize}
\item[i)] $\psi \in \fWAP(\RR^d)$ if and only if $\psi \in \cSp_\infty(\RR^d)$ and for all $f \in \cD$ we have $\psi*f \in WAP(\RR^d)$.
\item[ii)] $\psi \in \fSAP$ if and only if $\psi \in \cSp_\infty(\RR^d)$ and for all $f \in \cD$ we have $\psi*f \in SAP(\RR^d)$.
\item[iii)] $\psi \in \fWAP_0$ if and only if $\psi \in \cSp_\infty(\RR^d)$ and for all $f \in \cD$ we have $\psi*f \in WAP_0(\RR^d)$.
\end{itemize}
\end{lemma}
\begin{proof} We prove $(i)$, the other two are identical.

The forward implication is obvious as $\fWAP \subset \cSp_\infty(\RR^d)$ and $\cD \subset \cS$.

For the reverse implication, let $f \in \cS$. By translation boundedness we know that $\psi*f \in \Cu$.
Now, let $g_\alpha \in \cD$ be an approximate identity for $(\Cu, *)$. Then we have $\psi*g_\alpha \in WAP(\RR^d)$ and hence, as $f \in L^1(\RR^d)$ we also have \cite{ARMA} $(\psi*g_\alpha)*f \in WAP(\RR^d)$. By Lemma \ref{triple conv} we thus have
\[
(\psi*f)*g_\alpha=(\psi*g_\alpha)*f \in WAP(\RR^d) \,.
\]
As $(\psi*f)*g_\alpha$ converges uniformly to $\psi*f$ and $WAP(\RR^d)$ is closed in $\Cu$, we get $\psi*f \in WAP(\RR^d)$
\end{proof}

We now show that every tempered distribution with measure Fourier transform is weakly almost periodic. This class will thus always contain the subclass of positive definite tempered distributions, and hence all of our autocorrelations.

\begin{theorem}\label{FT are wap} Let $\psi \in \cSp$. If $\widehat{\psi}$ is a measure, then $\psi \in \fWAP(\RR^d)$.
\end{theorem}
\begin{proof}  Let $f \in \cS$. Then, by Proposition \ref{convolution is CU}, the function $\psi *f$ is the inverse Fourier transform of the finite measure $\widehat{f}\widehat{\psi}$.
The claim now follows from the fact that the Fourier transform of a finite measure is a weakly almost periodic function (see \cite{EBE2}).
\end{proof}

\begin{corollary}\label{PD are wap}
Let $\psi \in \cSp(\RR^d)$ be positive definite. Then $\psi \in  \fWAP(\RR^d)$.
\end{corollary}

We are now ready to construct the Eberlein decomposition for the space of weakly almost periodic distributions.

\begin{theorem}\label{WAP decomposition}
\[
\fWAP(\RR^d)=  \fSAP(\RR^d)\bigoplus \fWAP_0(\RR^d) \,.
\]
For $\psi\in\fWAP(\RR^d)$, we will denote this decomposition by
\[
\psi=\psi_{\mathsf{s}}+\psi_0 \,.
\]
Moreover, for all $f \in \cS$ we have
\[
\left(\psi*f \right)_{\mathsf{s}}= \psi_{\mathsf{s}}*f \;;\;\; \left(\psi*f \right)_{0}= \psi_{0}*f \,.
\]
\end{theorem}
\begin{proof}
For $\psi \in \fWAP(\RR^d)$ and $f \in \cS$, the function $\psi * (f_{\n}) \in WAP(\RR^d)$ and we may decompose $\psi * (f_{\n})$ into the functions
$\left( \psi * (f_{\n})\right)_{\mathsf{s}} \in SAP(\RR^d)$ and $\left( \psi * (f_{\n})\right)_{\mathsf{0}} \in WAP_0(\RR^d)$.

Therefore, for fixed $\psi\in\fWAP(\RR^d)$, the following functions are continuous:
\[ \cS(\RR^d) \to WAP(\RR^d) \to SAP(\RR^d) \to \CC \]
\[ f \mapsto \psi * (f_{\n}) \mapsto \left( \psi * (f_{\n})\right)_{\mathsf{s}} \mapsto \left( \psi * (f_{\n})\right)_{\mathsf{s}}(0) \]
and
\[ \cS(\RR^d) \to WAP(\RR^d) \to WAP_0(\RR^d) \to \CC \]
\[ f \mapsto \psi * (f_{\n}) \mapsto \left( \psi * (f_{\n})\right)_{\mathsf{0}} \mapsto \left( \psi * (f_{\n})\right)_{\mathsf{0}}(0)\,. \]
Therefore, we can define continuous mappings $S, T : \cS(\RR^d) \to \CC$ via
\[
T(f)=\left( \psi * (f_{\n})\right)_{\mathsf{s}}(0)
\]
and
\[
S(f)=\left( \psi * (f_{\n})\right)_{0}(0)\,.
\]
These mappings define two tempered distributions $\psi_{\mathsf{s}}$ and $\psi_0$.
It follows immediately from the construction that
\[
\psi=\psi_{\mathsf{s}}+\psi_0
\]
and that for all $f\in\cS$,
\[
\left(\psi*f \right)_{\mathsf{s}}= \left(\psi\right)_{\mathsf{s}}*f \,;\, \left(\psi*f \right)_{0}= \left(\psi\right)_{0}*f \,.
\]
To complete the proof, we need to show the uniqueness of the decomposition. For this it suffices to prove that
\[\fSAP(\RR^d) \cap \fWAP_0(\RR^d)=\{ 0 \} \,.\]
This follows immediately from the fact that $SAP(\RR^d) \cap WAP_0(\RR^d)=\{ 0 \}$ (see \cite{EBE3}).

If $\psi \in \fSAP(\RR^d) \cap \fWAP_0(\RR^d)$, then for all $f \in \cS$ we have $\psi *f \in  SAP(\RR^d) \cap WAP_0(\RR^d)$ and therefore $\psi *f =0$. In particular $\psi(f_{\n})=0$ for all $f \in \cS$.
\end{proof}

\begin{proposition} Let $\psi, \varphi \in \cSp$ be such that $\varphi$ has compact support.
\begin{itemize}
\item[(i)] If $\psi \in \fWAP(\RR^d)$ then $\psi * \varphi \in \fWAP(\RR^d)$.
\item[(ii)] If $\psi \in \fSAP(\RR^d)$ then $\psi * \varphi \in \fSAP(\RR^d)$.
\item[(iii)]If $\psi \in \fWAP_0(\RR^d)$ then $\psi * \varphi \in \fWAP_0(\RR^d)$.
\end{itemize}
\end{proposition}
\begin{proof}
The claim follows immediately from the observation that for $f \in \cS$, $\varphi\in\cSp$ with compact support, we have by Lemma \ref{conv compact supp is Schwartz} that $\varphi*f \in \cS$ and
\[
(\psi*\varphi)*f = \psi*( \varphi*f) \,.
\]
\end{proof}

\begin{example}
Let $\alpha$ be a multi-index on $\RR^d$ and consider the tempered distribution
\[
\psi :=\sum_{x \in \ZZ^d} D^{\alpha}\delta_z = D^{\alpha}\delta_{\ZZ^d} = D^{\alpha}\delta_0 \ast \delta_{\ZZ^d} \,.
\]
As $D^{\alpha}\delta_0$ has compact support and $\delta_{\ZZ^d}\in \fSAP(\RR^d)$, we see that $\psi \in \fSAP(\RR^d)$.
\end{example}

\begin{corollary}
Let $\alpha$ be a multi-index on $\RR^d$.
\begin{itemize}
\item[(i)] If $\psi \in \fWAP(\RR^d)$ then $D^\alpha\psi  \in \fWAP(\RR^d)$.
\item[(ii)] If $\psi \in \fSAP(\RR^d)$ then $D^\alpha\psi \in \fSAP(\RR^d)$.
\item[(iii)]If $\psi \in \fWAP_0(\RR^d)$ then $D^\alpha\psi  \in \fWAP_0(\RR^d)$.
\end{itemize}
\end{corollary}

We complete the section by proving that for translation bounded measures the notions of almost periodicity as measures and tempered distributions coincide.

\begin{theorem} Let $\mu \in \cM(\RR^d)$. Then
\begin{itemize}
\item[(i)] $\mu \in \fWAP (\RR^d) \Leftrightarrow \mu \in \mathcal{WAP}(\RR^d)$.
\item[(ii)] $\mu \in \fSAP (\RR^d) \Leftrightarrow \mu \in \mathcal{SAP}(\RR^d)$.
\item[(iii)]$\mu \in \fWAP_0 (\RR^d) \Leftrightarrow \mu \in \mathcal{WAP}_0(\RR^d)$.
\end{itemize}
\end{theorem}
\begin{proof}

The proof follows immediately from the density of $\cD$ in both $\Cc$ and $\cS$.
Since the measure $\mu$ is translation bounded, it is also translation bounded as a tempered distribution and we have
\[
\mu*f \in \Cu \, \mbox{ for all }\, f \in \Cc \cup \cS\,.
\]
We now prove (i). The other two results can be shown identically.
We start by picking some $\{g_\alpha\} \subseteq \cD$ which is an approximate identity for $(\Cu, * )$.

($\Rightarrow$) Let $f \in \cS$.
As $f \in L^1(\RR^d)$ we have $\mu * f \in  \mathcal{WAP}(\RR^d)$ \cite[Cor.~5.1]{ARMA}.
Moreover, as $ \mu*f \in \Cu$ we get $\mu * f \in  \mathcal{WAP}(\RR^d) \cap \Cu = WAP(\RR^d)$ \cite[Cor.~5.2]{ARMA}. This proves the claim.

($\Leftarrow$) Let $f \in \Cc$. As $g_\alpha*\mu \in WAP(\RR)$ and $f \in \Cc$ we have, for all $\alpha$, that
\[
(g_\alpha*\mu)* f \in WAP(\RR^d)
\]
Therefore $g_\alpha*(\mu* f) \in WAP(\RR^d)$.
We know that $\mu*f \in \Cu$. Therefore  $\{g_\alpha*(\mu* f)\}$ converges uniformly to $\mu*f$. As $g_\alpha*(\mu* f) \in WAP(\RR^d)$, and $WAP(\RR^d)$ is closed in $\Cu$, it follows that $\mu*f \in WAP(\RR^d)$.
\end{proof}

\section{Tempered distributions with measure Fourier transform}

To consider the diffraction of a tempered distribution, we take, as usual, the Fourier transform of the autocorrelation. As we have stated, if the autocorrelation exists, it is a positive definite tempered distribution, and thus its Fourier transform is a positive measure (the diffraction measure).

Our primary goal is to relate the Lebesgue decomposition of the diffraction measure to the Eberlein decomposition of the autocorrelation, and we do this in this section. In fact, the positivity of the diffraction measure plays no role in our proofs. Thus we here consider the larger class of tempered distributions whose Fourier transform is a measure.

\begin{definition} We define $\cMS$ to be the space of tempered distributions whose Fourier Transform is a measure, that is:
\[
\cMS := \{ \psi \in \cSp : \widehat{\psi} \mbox{ is a measure } \} \,.
\]
\end{definition}

\begin{remark}
\begin{itemize}
\item[(i)] Any positive definite tempered distribution is in $\cMS$. The connection between positive definiteness and the space $\cMS$ is explained by Proposition \ref{PD and measure FT} or by the Bochner Schwartz theorem.
\item[(ii)] By Theorem \ref{FT are wap} we have
\[
\cMS \subseteq \fWAP(\RR^d) \,.
\]
\end{itemize}
\end{remark}

As is the case for measures, the discrete and continuous parts of the Fourier transform of a positive definite distribution are exactly the Fourier transforms of the strongly and null weakly almost periodic parts of the distribution respectively.

\begin{theorem}\label{FT}
Let $\psi \in \cMS$. Then
\[
\left(\widehat{\psi}\right)_{\mathsf{pp}} =\widehat{\psi_{\mathsf{s}}} \, \mbox{ and }\, \left(\widehat{\psi}\right)_{\mathsf{c}} =\widehat{\psi_0}  \,.
\]
\end{theorem}
\begin{proof}
Let $f \in \cS$. Then, by Proposition \ref{convolution is CU}, $\widehat{\psi}\widehat{f}$ is a finite measure whose inverse Fourier transform is $f*\psi$. Accordingly, the pure point and continuous parts of $\widehat{\psi}\widehat{f}$, $\left( \widehat{f} \widehat{\psi}\right)_{\mathsf{pp}}$ and
$\left( \widehat{f} \widehat{\psi}\right)_{\mathsf{c}}$, are finite measures. Then by \cite{EBE2}, the inverse Fourier transforms of
$\left( \widehat{f} \widehat{\psi}\right)_{\mathsf{pp}}$ and $\left( \widehat{f} \widehat{\psi}\right)_{\mathsf{c}}$ are respectively strongly and null weakly almost periodic functions.

As
\[
f*\psi = \left[ \left( \widehat{f} \widehat{\psi}\right)_{\mathsf{pp}}\right]^\vee+\left[\left( \widehat{f} \widehat{\psi}\right)_{\mathsf{c}} \right]^\vee \,,
\]
by the uniqueness of the Eberlein decomposition and Theorem \ref{WAP decomposition} we must have

\begin{equation}\label{EQ2}
\begin{aligned}
\left[ \left( \widehat{f} \widehat{\psi}\right)_{\mathsf{pp}}\right]^\vee&=(f*\psi)_{\mathsf{s}}=f*(\psi_{\mathsf{s}})\,; \\
\left[\left( \widehat{f} \widehat{\psi}\right)_{\mathsf{c}} \right]^\vee &= (f*\psi )_0=f*(\psi_0)\,.
\end{aligned}
\end{equation}

Therefore, as $\left( \widehat{f} \widehat{\psi}\right)_{\mathsf{pp}}= \widehat{f} \left( \widehat{\psi}_{\mathsf{pp}}\right)$ and $\left( \widehat{f} \widehat{\psi}\right)_{\mathsf{c}}= \widehat{f} \left( \widehat{\psi}_{\mathsf{c}}\right)$, by taking the Fourier transforms we get
\begin{eqnarray*}
\begin{split}
\widehat{f} \left(  \widehat{\psi}_{\mathsf{pp}}\right) &=\widehat{f*(\psi_{\mathsf{s}})}=\widehat{f} \widehat{\psi_{\mathsf{s}}}\,; \\
\widehat{f} \left(  \widehat{\psi}_{\mathsf{c}}\right)&=\widehat{f*(\psi_0)}=\widehat{f} \widehat{\psi_0} \,.
\end{split}
\end{eqnarray*}
This proves the claim.
\end{proof}

Theorem \ref{FT} can be used to obtain a simple characterisation of tempered distributions with a pure point measure as a Fourier transform. While the characterisation is via test functions in $\cS$, we can easily extend this to $\cD$ and obtain a condition which is easier to test.

\begin{proposition}\label{SAP characterisation}
Let $\psi \in \cMS$. Then the following are equivalent:
\begin{enumerate}
\item[(i)] $\widehat{\psi}$ is a discrete measure.

\item[(ii)] $\psi \in \fSAP(\RR^d)$.
\item[(iii)] $\psi*f \in SAP(\RR^d)$ for all $f \in \cD$.
\item[(iv)]  $\psi*f*g \in SAP(\RR^d)$ for all $f,g \in \cD$.
\item[(v)]  $\psi*f*\widetilde{f} \in SAP(\RR^d)$ for all $f \in \cD$.
\end{enumerate}
\end{proposition}
\begin{proof}
$(i) \Leftrightarrow (ii)$ follows immediately from Theorem \ref{FT} and the uniqueness of the Eberlein decomposition of Theorem \ref{WAP decomposition}.

$(ii) \Rightarrow (iii)$ is obvious as $\cD \subseteq \cS$.

$(iii) \Rightarrow (iv) \Rightarrow (v)$ are obvious.

$(v) \Rightarrow (i)$ : Let $f\in\cD$. Since $\widehat{\psi}$ is a tempered measure, it follows that  $\left| \widehat{f} \right|^2 \widehat{\psi}$ is a finite measure. As the Fourier transform of this measure is $(\psi*f*\widetilde{f})_- \in SAP(\RR^d)$, it follows \cite{EBE2} that $\left| \widehat{f} \right|^2 \widehat{\psi}$ is a pure point measure.

As for each $R >0$ we can find some $f \in \cD$ which is non-vanishing on $B_R(0)$, it follows that the restriction of $\widehat{\psi}$ to $B_R(0)$ is pure point for each $R >0$. This implies that $\widehat{\psi}$ is pure point.
\end{proof}

Note that in the situation of Proposition \ref{SAP characterisation}, Theorem \ref{pure point intensity} will give us a way of describing the pure point measure $\widehat{\psi}$.

In exactly the same way, we can prove the following.

\begin{proposition}\label{WAP0 characterisation}
Let $\psi \in \cMS$. Then the following are equivalent:
\begin{itemize}
\item[(i)] $\widehat{\psi}$ is a continuous measure.
\item[(ii)] $\psi \in \fWAP_0(\RR^d)$.
\item[(iii)] $\psi*f \in WAP_0(\RR^d)$ for all $f \in \cD$.
\item[(iv)]  $\psi*f*g \in WAP_0(\RR^d)$ for all $f,g \in \cD$.
\item[(v)]  $\psi*f*\widetilde{f} \in WAP_0(\RR^d)$ for all $f \in \cD$.
\end{itemize}
\end{proposition}

As the proof is almost identical to the one of Proposition \ref{SAP characterisation}, we omit it.

\section{The mean of a weakly almost periodic distribution}

Fourier analysis is a very powerful tool in the study of $L^1$ functions and finite measures. If we want to extend some of the ideas of this theory to bounded functions and translation bounded measures, we cannot use integration anymore, as the integrals will be infinite. In any case, the translation boundedness of the measures, or boundedness of the functions, implies that we may still be able to integrate on average, meaning dividing the integral over $B_R$ by the volume of $B_R$ and taking the limit as $R\to\infty$. We will refer to this limit, when it exists, as the mean of the function.

As Eberlein showed \cite{EBE}, for weakly almost periodic functions the mean always exists and is uniform in translates of the function. Moreover,
\[
\langle f,g\rangle := M(f \bar{g})
\]
defines a semi-inner product on the space of weakly almost periodic functions, which becomes an inner product when we restrict to the class of strongly almost periodic functions. The space $(SAP(\RR^d), \langle  \cdot , \cdot \rangle)$ becomes a Hilbert space, and the set $\{ \chi_x(y) =e^{2\pi i x \cdot y} : x \in \RR^d \}$ of characters on $\RR^d$ is an orthogonal basis which is complete for this inner product. For more details we refer the reader to \cite{MoSt}.

Eberlein also showed that given a finite measure $\mu$, the function $f= \check{\mu}$ is weakly almost periodic, and the (Fourier-Bohr) coefficients $a_x(f) M(\overline{\chi_x} f)$ of $f$ with respect to this complete orthogonal set are exactly the intensities of the atoms in the pure point part $(\mu)_{pp}$ of the measure $\mu$.

The notion of the mean was extended to weakly almost periodic measures by Argabright and deLamadrid \cite{ARMA}. They further showed that for a Fourier transformable measure, the Fourier Bohr coefficients of the measure give exactly the intensity of the atoms in the pure point part of the measure's Fourier transform.

In this section we use the ideas of \cite{ARMA} to extend the notion of mean to weakly almost periodic distributions and prove that for distributions with measure Fourier transform, the relation between the Fourier Bohr coefficients and the pure point part of the Fourier transform still holds.

\begin{theorem}\label{WAP are amenable}
Let $\psi \in \fWAP(\RR^d)$. Then, there exists a constant $M(\psi)$ such that for all $f \in \cD$ we have
\[
M(\psi\ast f)=M(\psi) \int_{\RR^d} f \dd\lambda\,.
\]
\end{theorem}
\begin{proof}
Let $f, g \in \cD$. As $\psi*f \in WAP(\RR^d)$ it follows from \cite{ARMA} that
\[
M((\psi*f)*g)=M(\psi*f) \int_{\RR^d} g \dd\lambda
\]
and, similarly, that
\[
M((\psi*g)*f)=M(\psi*g) \int_{\RR^d} f \dd\lambda \,.
\]
This implies that
\[
M(\psi*f) \int_{\RR^d} g\dd\lambda = M(\psi*g) \int_{\RR^d} f \dd\lambda\,,
\]
and we see that the ratio $\ds\frac{M(\psi*f)}{\int_{\RR^d} f\dd\lambda}$ is constant over all $f\in\cD$ with non-zero integral. This allows us to define the number
\[
M(\psi) := \frac{M(\psi*f_0)}{\int_{\RR^d} f_0\dd\lambda }\,,
\]
where $f_0 \in \cD$ is an arbitrary function with $\int_{\RR^d} f_0 \dd\lambda \neq 0$.
%Then the definition of $M(\psi)$ doesn't depend on the choice of $f_0$.

Finally, if $g \in \cD$ is any function, we get
\[
 M(\psi*g) \int_{\RR^d} f_0 \dd\lambda = M(\psi*f_0) \int_{\RR^d} g \dd\lambda = M(\psi)\int_{\RR^d} f_0 \dd\lambda \int_{\RR^d}  g \dd\lambda   \,.
\]
Dividing both sides by $\int_{\RR^d} f_0 \dd\lambda \neq 0$ we obtain the required identity.
\end{proof}

The constant found in the above theorem is called the mean of $\psi$.

\begin{definition}
Let $\psi \in \fWAP(\RR^d)$. The {\em mean} of $\psi$ is the number $M(\psi)$, defined by
\[
M(\psi*f) = M(\psi) \int_{\RR^d} f(t) d \lambda(t) \, \mbox{ for all }\, f \in \cD \,.
\]
\end{definition}

We next show that Theorem \ref{WAP are amenable} can be extended to test functions $f \in \cS$.

\begin{proposition} Let $\psi \in \fWAP(\RR^d)$. Then for all $f \in \cS$ we have
\[
M(\psi\ast f)=M(\psi) \int_{\RR^d} f \dd\lambda \,.
\]
\end{proposition}
\begin{proof}
Let $g \in \cD$ with $\int_{\RR^d} g\dd\lambda \neq 0$ and let $f\in\cS$. Then $f*g \in \cS$ and, as $\psi \in \fWAP(\RR^d)$, we have
\[
\psi*f \in WAP(\RR^d) \,;\, \psi*g \in WAP(\RR^d)  \mbox{ and } \, \psi*(f*g) \in WAP(\RR^d)\,.
\]
Then by \cite{ARMA},
\[
M((\psi*f )*g) = M(\psi*f) \int_{\RR^d} g d \lambda \,,
\]
and thus, by Lemma \ref{triple conv},
\[
 M(\psi*f) \int_{\RR^d} g d \lambda =M((\psi*f )*g) =M(\psi*(f *g))=M(\psi*(g *f))= M((\psi*g)*f)\,.
\]
Now, $\psi*g \in WAP(\RR^d)$ and the convolution $(\psi*g)*f$ of functions can be interpreted as the convolution between the measure $(\psi*g)\lambda$ and the finite measure $f \lambda$. Therefore, by \cite[Prop.~4.3]{ARMA} we have
\[
M((\psi*g)*f)= M(\psi*g) \int_{\RR^d} f d \lambda \,.
\]
The claim follows now from Theorem \ref{WAP are amenable}.
\end{proof}

\begin{corollary} Let $\psi \in \fWAP(\RR^d)$ and $f \in \cS$. Then
\[
M(\psi) \int_{\RR^d} f d \lambda = \lim_n \frac{1}{(2n)^d} \int_{[-n,n]^d} \psi ( T_t f_-) dt  \,.
\]
\end{corollary}

As for functions and measures, the mean is linear and is compatible with the operatons of reflection and conjugation.

\begin{proposition} The mean $M : \fWAP(\RR^d) \to \CC$ has the following properties:
\begin{enumerate}
\item[(i)] For all $\psi, \phi\in\fWAP(\RR^d)$ and $a,b\in\CC$, $M(a\psi+b\phi) =a M(\psi)+bM(\phi)$.
\item[(ii)] $M(\lambda_{\RR^d}) =1$.
\item[(iii)] If $f \in WAP(\RR^d) \cap \cS$ then the means of $f$ as function and distribution are equal.
\item[(iv)] If $\mu \in {\mathcal WAP}(\RR^d) \cap \cSp$ then the means of $\mu$ as measure and distribution are equal.
\item[(v)] For $\psi\in\fWAP(\RR^d)$, $M(\overline{\psi})=M(\widetilde{\psi})=\overline{ M(\psi)}$ and $M(\psi_{\n})= M(\psi)$.
\item[(vi)] For $\psi\in\fWAP(\RR^d)$, $x\in\RR^d$, $M(T_x\psi)= M(\psi)$.
\end{enumerate}
\end{proposition}
\begin{proof}
(i) and (ii) follow from \cite[Theorem 14.1]{EBE} and Theorem \ref{WAP are amenable}.

(iii): If $f \in WAP(\RR^d) \cap \cS$ and $g \in \cD$ then as functions we have \cite{ARMA}
\[
M(f*g) =M(f) \int_{\RR^d} g d \lambda \,.
\]
The claim follows now from Theorem \ref{WAP are amenable}.

(iv): If $\mu \in {\mathcal WAP}(\RR^d) \cap \cS'$ and $g \in \cD$ then the convolution $\mu*g$ of the measure $\mu$ and the compactly supported continuous function $g$ satisfies \cite{ARMA}
\[
M(\mu*g) =M(\mu) \int_{\RR^d} g d \lambda \,.
\]
As $\mu*g$ is the same if we consider it as the convolution of a tempered distribution and a Schwartz function, the claim follows again from Theorem \ref{WAP are amenable}.

(v) and (vi) follow from \cite[Theorem 14.1]{EBE} and Theorem \ref{WAP are amenable}.
\end{proof}

Recall that for each $x \in \RR^d$  and $\psi \in \cS$ we have
\[
\widehat{T_x \psi} =\chi_{-x} \widehat{\psi} \; \mbox{ and } \; \widehat{\chi_x \psi} =T_{x} \widehat{\psi}
\]
where for $x,y\in\RR^d$,
\[
\chi_x(y) =e^{2\pi i x \cdot y} \,.
\]

The importance of the mean for diffraction theory is given by the following type of results, see also \cite{EBE2,ARMA,MoSt,hof95}.

\begin{theorem}\label{pure point intensity} Let $\psi \in \cMS$. Then, for all $x \in \RR^d$ we have
\[
\widehat{\psi}(\{ x \}) = M(\overline{\chi_x} \psi) \,.
\]
\end{theorem}
\begin{proof}
Let $f \in \cD$. Then, $\widehat{f} \widehat{\psi}$ is a finite measure, and hence, by \cite{EBE2} we have
\[
\left(\widehat{f}\widehat{\psi}\right)(\{ 0 \}) = M(\psi\ast f)=M(\psi)\int_{\RR^d} f\dd\lambda \,.
\]
As
\[
\widehat{f}(0) =\int_{\RR^d} f \dd\lambda\,,
\]
it follows that $\widehat{\psi}(\{0\}) = M(\psi)$; that is, our claim is true for $\chi=0$.

For general $x\in\RR^d$, we then have
\[
\widehat{\psi}(\{ x \})=T_{-\chi_x}\widehat{\psi}(\{ 0 \}) =\widehat{(\overline{\chi_x}\psi)}(\{ 0 \})= M(\chi_x \psi) \,.
\]
\end{proof}

Exactly as in the case of functions and measures, we can characterize the null weakly almost periodic tempered distributions in terms of means under character multiplication.

\begin{proposition} Let $\psi \in \fWAP(\RR^d)$. Then $\psi \in \fWAP_0(\RR^d)$ if and only if for all $x \in \RR^d$ we have
\[
M(\chi_x \psi) =0 \,.
\]
\end{proposition}
\begin{proof}
Note that for $f \in \cS\,, x\in\RR^d$, we have
\[
((\chi_x \psi)*f)\widehat{\phantom{x}} = \left( T_x \widehat{\psi} \right) \widehat{f} = T_x \left(  \widehat{\psi} T_{-x}\widehat{f} \right)\,,
\]
and thus
\begin{equation}\label{eq chi conv}
(\chi_x \psi)*f = \chi_x \left( \psi*( \chi_{-x} f)  \right)\,.
\end{equation}

Let $\psi \in \fWAP_0 (\RR^d)$. For all $f \in \cD$ and $x \in \RR^d$, we have $(\psi\ast(\chi_{-x} f)) \in WAP_0(\RR^d)$, and hence by \cite{EBE3},
\[
M\left(\chi_x ( \psi*( \chi_{-x} f) )  \right) =0 \,.
\]
Then from (\ref{eq chi conv}), for all $f \in \cD$ we have
\[
M( (\chi_x \psi)*f ) =0 \,,
\]
so by the definition of the mean, we have the forward implication.

Conversely, suppose that $M(\chi_x \psi) =0$ for all $x \in \RR^d$. Let $f \in \cS$. Then
\[
M \left( (\chi_x \psi)*(\chi_x f \right) ) = M( \chi_x \psi ) \int_{\RR^d} \chi_x f d \lambda =0 \,
\]
for all $x \in \RR^d$. Applying (\ref{eq chi conv}), this yields that
\[
0 = M\left(\chi_x ( \psi*( \chi_{-x} \chi_x f) ) \right) = M( \chi_x ( \psi* f) )
\]
for all $x \in \RR^d$ and hence \cite{EBE3} we have that $\psi *f \in WAP_0(\RR^d)$.
As $f \in \cS$ was arbitrary, we obtain the result.
\end{proof}

Note that in case $\psi \in \cMS$, the above result is an immediate corollary of Theorem \ref{pure point intensity}.

Combining the results of this section with Proposition \ref{WAP0 characterisation}, we obtain the following.

\begin{proposition}\label{WAP0 characterisation 2}
Let $\psi \in \cMS$. Then the following are equivalent:
\begin{enumerate}
\item[(i)] $\widehat{\psi}$ is a continuous measure.
\item[(ii)] $\psi \in \fWAP_0(\RR^d)$.
\item[(iii)] $\psi*f \in WAP_0(\RR^d)$ for all $f \in \cD$.
\item[(iv)]  $\psi*f*g \in WAP_0(\RR^d)$ for all $f,g \in \cD$.
\item[(v)]  $\psi*f*\widetilde{f} \in WAP_0(\RR^d)$ for all $f \in \cD$.
\item[(vi)]  $M(\chi_x \psi) =0 $ for all $x \in \RR^d$.
\end{enumerate}
\end{proposition}

\section{Diffraction of a Weakly Almost Periodic Distribution}

In this section we study the diffraction of a weakly almost periodic distribution. We shall see that as long as such a distribution admits an autocorrelation, the diffraction is unique, pure point diffractive and can be expressed in terms of the Fourier-Bohr coefficients of the original distribution.

\begin{theorem}\label{WAP are pure point diffractive}
Let $\psi \in \fWAP$ which has an autocorrelation $\phi$. Then the following hold.
\begin{enumerate}
\item[(i)] For all $f \in \cD$ we have
\[
\phi*f*\widetilde{f} = (\psi*f) \econv \widetilde{\psi*f}
\]
where the RHS denotes the Eberlein convolution of weakly almost periodic functions \footnote{See \cite[Def.~15.1]{EBE} for the definition of Eberlein convolution of weakly almost periodic functions.}.
\item[(ii)] $\phi$ is the only autocorrelation of $\psi$.
\item[(iii)] If $\{ h_R \}$ is any strong smooth approximate van Hove family for $\psi$ then
\[
\phi = \lim_R \frac{1}{\vol(B_R)} \psi_R * \widetilde{\psi_R}
\]
\item[(iv)] $\phi \in \fSAP(\RR^d)$.
\item[(v)] $\psi$ has pure point diffraction given by
\[
\widehat{\phi}= \sum_{x \in \RR^d} \left| M(\overline{\chi_x}\psi) \right|^2 \delta_x \,.
\]
\end{enumerate}

\end{theorem}
\begin{proof}

In this proof we follow closely the proofs of \cite{LS2}.

(i): Let $R_n \to \infty$ be such that
\[
\phi= \lim_n \frac{1}{\vol(B_{R_n})} \psi_{R_n}*\widetilde{\psi_{R_n}}\,.
\]

Let $f \in \cD$ and let $R>0$ be such that $\sup(f) \subseteq B_R(0)$. Let $N$ be such that for all $n >N$ we have $R_n>R$.

For any $n > N$, we note that
\[
(\psi_{R_n} *f)(t) - h_{R_n}(t) (\psi*f)(t) = \psi ( h_{R_n} T_t f_- ) -h_{R_n}(t) \psi (T_t f_- ) \,.
\]
If $t \notin B_{R_n+1+R}(0)$ then $h_{R_n} T_t f_- \equiv 0$ and hence
\[
 \psi ( h_{R_n} T_t f_- ) =0 \,.
\]
In this case we also have $h_{R_n}(t)=0$. This shows that
\[
(\psi_{R_n} *f)(t) = h_{R_n}(t) (\psi*f)(t)  \, \mbox{ for } t \notin B_{R_n+1 +R}(0) \,.
\]
Now, taking $n>N$, if $t \in B_{R_n-R}(0)$ we have
\[
h_{R_n} T_tf_-= T_tf_- \,,
\]
and hence
\[
(\psi_{R_n} *f)(t)= \psi ( h_{R_n}  T_t f_- )= \psi ( T_t f_- )=h_{R_n}(t) \psi (T_t f_- )  =h_{R_n}(t) (\psi*f)(t) \,.
\]
This shows that
\[
(\psi_{R_n} *f)(t) = h_{R_n}(t) (\psi*f)(t)  \, \mbox{ for all } t\notin B_{R_n+1+R}(0) \backslash B_{R_n-R}(0) \,.
\]
Moreover, on the compact set $\overline{B_{R_n+1+R}(0)} \backslash B_{R_n-R}(0)$, the function $\psi_{R_n} *f = h_{R_n} (\psi*f)$ is bounded as $\| \psi*f \|_\infty < \infty$ by translation boundedness of $\psi$. Therefore, by a standard van Hove computation, we have
\begin{align*}
\lim_{n} \frac{1}{\vol(B_{R_n})} &\;(\psi_{R_n}\ast f)*(\widetilde{\psi_{R_n}}*\widetilde{f}) (t) \\
&- \lim_{R_n} \frac{1}{\vol(B_{R_n})} (h_{R_n} (\psi*f))*(h_{R_n} (\psi*f))\widetilde{\phantom{x}}(t) \; = \; 0 \,,
\end{align*}
for all $t \in \RR^d$. Moreover, by the properties of $h_{R_n}$ and the van Hove property of $B_{R_n}$, we have
\begin{align*}
\lim_{n \to \infty} \frac{1}{\vol(B_{R_n})} &(h_{R_n} (\psi*f))* (h_{R_n} (\psi*f))\widetilde{\phantom{x}} (t) \\
&- \frac{1}{\vol(B_{R_n})} (1_{B_{R_n}} (\psi*f))*(1_{B_{R_n}} (\psi*f))\widetilde{\phantom{x}} (t) \; = \; 0\,,
\end{align*}
for all $t \in \RR^d$. Therefore,
\begin{equation}\label{EQ3}
\begin{aligned}
\lim_{n \to \infty} \frac{1}{\vol(B_{R_n})} &(\psi_{R_n}\ast f)*\widetilde{\psi_{R_n}}*\widetilde{f} (t) \\
&- \frac{1}{\vol(B_{R_n})} (1_{B_{R_n}} (\psi*f))*(1_{B_{R_n}} (\psi*f))\widetilde{\phantom{x}} (t) \; = \; 0 \,,
\end{aligned}
\end{equation}
for all $t \in \RR^d$.
Now, since $\psi \in \fWAP(\RR^d)$ we have $\psi*f \in WAP(\RR^d)$. Therefore, the Eberlein convolution $(\psi*f)\econv \widetilde{(\psi*f)}$ is well defined and by \cite{EBE},
\[
(\psi*f)\econv \widetilde{(\psi*f)}(t) = \lim_n \frac{1}{\vol(B_{R_n})} (1_{B_{R_n}} (\psi*f))*(1_{B_{R_n}} (\psi*f))\widetilde{\phantom{x}} (t)
\]
uniformly in $t$. Therefore, by (\ref{EQ3}), for all $t \in \RR^d$ the limit
\[
\lim_{n \to \infty} \frac{1}{\vol(B_{R_n})} (\psi_{R_n}\ast f)*(\widetilde{\psi_{R_n}}*\widetilde{f}) (t)
\]
exists and is equal to $(\psi*f)\econv \widetilde{(\psi*f)}(t)$.

This proves (i).

(ii): Let $\varphi$ be any other autocorrelation of $\psi$. Then, as $(i)$ holds for any autocorrelation, we get that for all $f \in \cD$ we have
\[
\phi*f*\widetilde{f} = (\psi*f) \econv \widetilde{\psi*f}=\varphi*f*\widetilde{f}
\]

Evaluating at $t =0$ we get that $\phi=\varphi$ on a dense subset of $\cS$, and hence $\phi=\varphi$.

(iii): This follows immediately from the compactness of $\{  \frac{1}{\vol(B_{R})} \psi_{R}*\widetilde{\psi_{R}} | R> 1 \}$ and from the fact that as $R \to \infty$ this set has a unique cluster point.

(iv): Since $(\psi*f)\econv \widetilde{(\psi*f)} \in SAP(\RR^d)$ \cite[Theorem ~15.1]{EBE}, we have from (i) that
$\phi*f*\widetilde{f} \in SAP(\RR^d)$. The claim now follows from Proposition \ref{SAP characterisation}.

(v): From (i), for each $f \in \cD$ we have
\[
\phi*f*\widetilde{f} = (\psi*f)\econv \widetilde{\psi*f} \,.
\]
As $\phi \in \fSAP(\RR^d)$ is positive definite, $\widehat{\phi}$ is a pure point measure by Proposition \ref{SAP characterisation}. To prove $(iv)$ we need to show that $\widehat{\phi}(\{x\})= \left| M(\chi_x \psi) \right|^2$ for all $x\in\RR^d$.

Fix $x \in \RR^d$ and take $f \in \cD$ such that $\widehat{f}(x) \neq 0$. We have
\[
\left| \widehat{f}(x) \right|^2 \widehat{\phi}(\{x\}) = (\phi*f*\widetilde{f})\widehat{\phantom{x}}(x)
      = M(\overline{\chi_x}(\phi*f*\widetilde{f}) )  = M( \overline{\chi_x} ((\psi*f)\econv \widetilde{(\psi*f)}))  \,,
\]
where we have applied Theorem \ref{pure point intensity} and (i).
By \cite[Lemma~15.2]{EBE} we know that the Fourier Stiltje coefficient $M(\overline{\chi_x}( (\psi*f)\econv \widetilde{(\psi*f)}))$ of the Eberlein convolution is given by
\[
M(\overline{\chi_x}(\psi*f)\econv\widetilde{(\psi*f)})= \left|M(\overline{\chi_x}(\psi*f))\right|^2
				= \left|M(\overline{\chi_x} \psi) \right|^2 \left| \widehat{f}(x) \right|^2 \,.
\]
This shows that
\[
\left| \widehat{f}(x) \right|^2 \widehat{\phi}(\{x\}) = \left|M(\overline{\chi_x} \psi) \right|^2 \left| \widehat{f}(x) \right|^2 \,,
\]
which completes the proof.
\end{proof}

\section{Tempered distributions with an autocorrelation}

Our primary goal for this section is to introduce few examples of tempered distributions for which the diffraction can be calculated explicitly.

\begin{example}\label{delta Z'}
Returning again to the tempered distribution of Example \ref{eg DdeltaZ}, let $\psi = -D\delta_{\ZZ}$ on $\RR$. Writing
\[ \psi = -D\delta_0 \ast \delta_{\ZZ} \]
and observing that $D\delta_0\ast D\delta_0 = D^2\delta_0$, a simple calculation gives that $\psi$ has autocorrelation
\[
\phi = D^2\delta_{\ZZ} = \sum_{n \in \ZZ} D^2 \delta_n
\]
and diffraction
\[
\widehat{\phi} = 4 \pi^2 \sum_{n \in \ZZ} n^2\delta_n  \,.
\]
This can be easily generalised from the Dirac measure of $\ZZ$ to that of a lattice $\Lambda\in \RR^n$. For a multi-index $\alpha$, let
\[ \psi := D^{\alpha}\delta_{\Lambda}\,.\]
Then $\psi$ has autocorrelation given by
\[ \phi := \dens(\Lambda) D^{2\alpha} \delta_{\Lambda}\,,\]
and hence diffraction
\[ \widehat{\phi} = (4\pi^2)^{|\alpha|} \dens(\Lambda)^2 \sum_{x\in \Lambda^*} x^{2\alpha} \delta_x\,.\]

We will see later in this section that the results of this paper will yield a simpler derivation of these autocorrelation and diffraction measures.
\end{example}

Next, we change a standard probabilistic model to give a new interesting example:

\begin{example} Construct a distribution $\psi$ the following way: for each integer $n \in \ZZ$ select at random $\delta_n$ or $\delta'_n$ with respective probabilities $p$ and $1-p$. This distribution can be written as
\[
\psi=\sum_{n \in \ZZ} a_n \delta_n + b_n \delta'_n
\]
where $a_n,b_n \in \{0,1 \}$ and $a_n+b_n=1$. An easy computation shows that the autocorrelation $\phi$ of $\psi$ has the form
\[
\phi=\sum_{n \in \ZZ} c_n \delta_n + d_n \delta'_n + e_n \delta''_n
\]
where
\[
c_n =\lim_m \frac{1}{2m} \sum_{k=-m}^m a_k a_{n+k}\,,
\]
\[
d_n =2\lim_m \frac{1}{2m} \sum_{k=-m}^m a_k b_{n+k}\,,
\]
\[
e_n =\lim_m \frac{1}{2m} \sum_{k=-m}^m b_k b_{n+k}\,.
\]

Now, if $n \neq 0$ we have, by the independence of the random variables, that $a_k$ and $a_{n+k}$ are both 1 with probability $p^2$, otherwise their product is zero. Therefore, we get $c_n =p^2$. In the same way, $d_n=2p(-1p)$ and $e_n=(1-p)^2$.

If $n=0$, we have $a_k=a_{n+k}$, $b_k=n_{n+k}$, and $a_kb_k=0$. Therefore we get $c_n =p, d_n=0$, and $e_n=1-p$. This gives
\[
\phi= p^2 \delta_\ZZ + 2p(1-p) D \delta_\ZZ + (1-p)^2 D^2 \delta_{\ZZ}+ (p-p^2) \delta_0+ (1-p-(1-p)^2) D^2\delta_0 \,.
\]

The diffraction of $\psi$ is hence
\[
\widehat{\phi}= \bigl( \sum_{n \in \ZZ} (p^2+2np(1-p)+n^2(1-p)^2) \delta_n \bigr) + \bigl( (p-p^2) + (1-p-(1-p)^2)x^2 \bigr) \lambda \,.
\]
\end{example}

The next result will give us many examples of distributions with pure point diffraction.

\begin{lemma}\label{L temp diff}  Let $\mu \in \cM^\infty(\RR^d)$ and let $\psi = \check{\mu} \in \cSp(\RR^d)$. If $\psi$ has an autocorrelation $\phi$, then $\phi$ is the unique autocorrelation of $\psi, \phi \in \fSAP(\RR^d)$ and the diffraction of $\psi$ is
\[
\widehat{\phi} = \sum_{x \in \RR^d} \bigl| \mu( \{ x \} ) \bigr|^2 \delta_x \,.
\]
\end{lemma}
\begin{proof}
As $\psi \in \cMS$ we have $\psi \in \fWAP(\RR^d)$. Therefore by Proposition \ref{WAP are pure point diffractive}, $\phi$ is the unique autocorrelation of $\psi, \phi \in \fSAP(\RR^d)$ and the diffraction of $\psi$ is
\[
\widehat{\phi} = \sum_{x \in \RR^d} \bigl| M(\bar{\chi_x}\psi) \bigr|^2 \delta_x \,.
\]

Moreover, as $\psi \in \cMS$ by Theorem \ref{pure point intensity} we have
\[
M(\bar{\chi_x}\psi)= \widehat{\psi}( \{ x \})= \mu( \{ x \})  \,.
\]

This completes the proof.
\end{proof}

\begin{corollary}\label{cor diff td}
Let $\psi \in \cMS$ be a tempered distribution with autocorrelation $\phi$. Then the diffraction of $\psi$ has the form
\[
\widehat{\phi}= \sum_{x \in \RR^d} \left| \widehat{\psi}(\{x \})\right|^2 \delta_x\,.
\]
\end{corollary}

\begin{corollary}\label{C temp diff}  Let $\Lambda \subseteq \widehat{\RR^d}$ be a Delone set and let $\psi = \check{\delta_{\Lambda}}$.  If $\psi$ admits an autocorrelation $\phi$ then $\phi$ is the only autocorrelation of $\psi$, the diffraction of $\psi$ is $\delta_{\Lambda}$ and $\psi =\phi$.
\end{corollary}
\begin{proof}

The first part follows from Lemma \ref{L temp diff}. The last claim follows from the fact that $\widehat{\psi -\phi}= \delta_\Lambda- \delta_\Lambda =0$.

\end{proof}

\begin{remark} If $\Lambda$ is repetitive, has FLC and is not fully periodic, we have $\delta_\Lambda \notin \mathcal{WAP}(\RR^d)$ \cite{LS2}. As the Fourier Transform of any translation bounded measure is a  weakly almost periodic measure \cite{ARMA,MoSt}, it follows that $\delta_{\Lambda}$ cannot in this case be the diffraction measure of any translation bounded measure.
\end{remark}

\begin{example} Let $\Lambda$ be the Fibonacci point set, and $\psi = \check{\delta_{\Lambda}}$. If $\psi$ has an autocorrelation, then the only  diffraction of $\psi$ is $\delta_\Lambda$.
\end{example}

For the next example the existence of the autocorrelation was established in \cite{ter}.
\begin{example}
Recall from Example \ref{eg aperiodic} the tempered distribution
\[ \widehat{\omega} = \tfrac{1}{2}\delta_{\frac{\ZZ}{2}} + \sum_{n\geq 1} \frac{\cos(2\pi(4^n-1)(\cdot))}{4^n}\delta_{\frac{\ZZ}{2.4^n}}\,,\]
whose Fourier transform is the measure
\[\omega := \delta_{2\ZZ} + \sum_{n\geq 1} \delta_{2.4^n \ZZ} \ast (\delta_{4^n - 1} + \delta_{1-4^n}) \,.\]
By Corollary \ref{cor diff td}, the diffraction of $\widehat{\omega}$ is exactly $\omega$ (and hence its autocorrelation is $\widehat{\omega}$).
\end{example}

We next look at how convolution with a compactly supported distribution affects the diffraction.

\begin{proposition} Let $\psi \in \cSp$ be translation bounded and $\varphi$ be a tempered distribution with compact support. If $\phi$ is an autocorrelation of $\psi$ calculated with respect to $h_{R_n}$ and if
\[
\left\{ \frac{1}{\vol(B_{R_n})} (h_{R_n} \psi*\vartheta)* \widetilde{h_{R_n} \psi*\vartheta} : n\in\NN \right\}
\]
is weak-* precompact, then then $\phi*\vartheta*\widetilde{\vartheta}$ is the autocorrelation of $\psi*\vartheta$ calculated with respect to $h_{R_n}$.
\end{proposition}
\begin{proof}

It is easy to see that convolution with a fixed compactly supported distribution is weak-* continuous. Therefore
\[
\phi*\vartheta*\widetilde{\vartheta}= \lim_{n \to \infty} \frac{1}{\vol(B_{R_n})} \bigl( (h_{R_n} \phi)*\widetilde{(h_{R_n} \phi)} \bigr) * *\vartheta*\widetilde{\vartheta}
\]

By a similar computation to Lemma \ref{L4.5} we get
\[
\lim_{n \to \infty} \frac{1}{\vol(B_{R_n})}\bigl[ \bigl( (h_{R_n} \phi)*\widetilde{(h_{R_n} \phi)} \bigr) *\vartheta*\widetilde{\vartheta}-\bigl( (h_{R_n} \phi*\vartheta)*\widetilde{(h_{R_n} \phi*\vartheta)} \bigr) \bigr] (f*g) =0
\]

for all $f,g \in \cD$. Now by our convolution assumption we get by Lemma \ref{conv on cD implies cS} that
\[
\phi*\vartheta*\widetilde{\vartheta}=\lim_{n \to \infty} \frac{1}{\vol(B_{R_n})}\bigl( (h_{R_n} \phi*\vartheta)*\widetilde{(h_{R_n} \phi*\vartheta)} \bigr)
\]
in the weak-* topology of $\cSp$, which proves our claim.
\end{proof}

An immediate consequence of this is the following.

\begin{corollary}\label{ac derivative} Let $\mu \in \cM$ be a translation bounded measure and let $h_R$ be a smooth approximate van Hove sequence as in Proposition \ref{derivative tb measures}. Let $h_{R_n}$ be a sequence with respect to which $D^\alpha \mu$ has an autocorrelation $\phi$ and $\mu$ has autocorrelation $\gamma$. Then
\[
\phi= D^{2 \alpha} \gamma
\]
and the diffraction of $\psi$ is thus
\[
\widehat{\phi}=(2\pi)^{2|\alpha|}x^{2 \alpha} \widehat{\gamma}\,.
\]
\end{corollary}

Let us note that in Corollary \ref{ac derivative}, the existence of a sequence $h_{R_n}$ is guaranteed by Proposition \ref{tb measures have strong}  and Proposition \ref{derivative tb measures}.

\begin{example} Let $\Lambda$ be any Delone set. If $\psi$ and $\gamma$ are autocorrelations of $D \delta_{\Lambda}$ and $\delta_{\Lambda}$ calculated with respect to the same choice of the smooth van Hove sequence, respectively the corresponding van Hove sequence, then
\[
\widehat{\psi} = x^2 \widehat{\gamma} \,.
\]

\end{example}

We complete the paper by introducing a non translation bounded tempered distribution which has an autocorrelation. Our example is a measure, and has also the same autocorrelation as measure.

\begin{example} Let
\[
\mu := \delta_\ZZ +\sum_{n \in \NN} n \delta_{2^n} \,.
\]
As $\mu$ has logarithmic growth, it is a tempered distribution. Also, $\mu$ is positive and not translation bounded as a measure, hence not translation bounded as a tempered distribution.

Now, if we pick any smooth approximate van Hove sequence $h_R$, we have for all $m \in \NN$
\[
h_m \mu = \delta_{\ZZ \cap [-m,m]} +\sum_{n \in \NN ; 2^n \leq m} n \delta_{2^n} \,.
\]
Therefore
\begin{eqnarray*}
\begin{split}
\frac{1}{2m} \bigl( ( h_m \mu ) * \widetilde{ h_m \mu} - \delta_{\ZZ \cap [-m,m]} *\widetilde{ \delta_{\ZZ \cap [-m,m]}} \bigr) &= \frac{1}{2m} \bigl( \sum_{k=-m}^m \sum_{n \in \NN; 2^n  \leq  m} n \delta_{k-2^n}+ n \delta_{2^n-k} \bigr) \\
&+\frac{1}{2m} \sum_{n,k \in \NN ; 2^n \leq m ; 2^k \leq m} nk \delta_{2^n-2^k} \,.
\end{split}
\end{eqnarray*}

Now, for each $l \in \ZZ$ there are at most $2\log_2(m)$ pairs $-k \leq k \leq m$ and $n \in \NN$ with $2^n <m$ such that $k-2^n=l$ or $2^n-k=l$. For each such pair, we also have $n \leq \log_2(m)$.

Moreover, using the binary representation of a number, it is easy to see that every integer $l$ can be written in at most one way as the difference $l=2^k-2^n$ of two powers of two. We also have $nk < (\log_2(m))^2$. Therefore we get
\begin{eqnarray*}
\begin{split}
0 \leq \frac{1}{2m} \bigl( ( h_m \mu ) * \widetilde{ h_m \mu} - \delta_{\ZZ \cap [-m,m]} *\widetilde{ \delta_{\ZZ \cap [-m,m]}} \bigr) & \leq \frac{1}{2m} \bigl( \sum_{l=-2m}^{2m}  3(\log_2(m))^2 \delta_l \bigr) \,.
\end{split}
\end{eqnarray*}

It is easy to see that both as measures and tempered distributions we have
\[
\lim_m \frac{1}{2m} \bigl( \delta_{\ZZ \cap [-m,m]} *\widetilde{ \delta_{\ZZ \cap [-m,m]}} \bigr) =\delta_\ZZ \,,
\]
and
\[
\lim_m \frac{1}{2m} \bigl( \sum_{l=-2m}^{2m}  3(\log_2(m))^2 \delta_l \bigr)=0 \,.
\]
Therefore
\[
\lim_m \frac{1}{2m} \bigl( ( h_m \mu ) * \widetilde{ h_m \mu} \bigr)  = \delta_\ZZ \,.
\]

\end{example}

{\large \sc \bf Acknowledgment:} {\small
The authors are grateful to Michael Baake for stimulating our interest in these questions, as well as for numerous discussions and comments which improved the quality of this manuscript. Part of the work was done while NS visited the University of Bielefeld and NS would like to thank the University for the hospitality. The work was partially supported by the German Research Foundation (DFG), within the CRC 701 and partially supported by NSERC with a research grant number 2014-03762 and the authors are grateful for the support.}

\end{document}